\documentclass[preprint]{imsart}

\usepackage{amsthm,amsmath,natbib}

\startlocaldefs

\usepackage{graphics,amssymb,color}
\usepackage{graphicx,subfigure}
\usepackage{algorithm}
\usepackage{algorithmic}

\newcommand{\argmax}{\mathop{\mathrm{argmax}}}
\newcommand{\argmin}{\mathop{\mathrm{argmin}}}
\newcommand{\minimize}{\mathop{\mathrm{minimize}}}

\newcommand{\st}{\mathop{\mathrm{subject\,\,to}}}

\newcommand{\real}{\mathbb{R}}
\newcommand{\Pp}{{\mathbb P}}
\newcommand{\Ee}{{\mathbb E}}
\newcommand{\Cov}{\mathrm{Cov}}
\newcommand{\sign}{\mathrm{sign}}
\newtheorem{theorem}{Theorem}
\newtheorem{lemma}{Lemma}

\newtheorem{remark}{Remark}

\newcommand{\K}{\mathcal{K}}
\newcommand{\V}{\mathcal{V}}
\newcommand\tf{\widetilde{f}}
\newcommand\tX{\widetilde{X}}
\newcommand\ty{\widetilde{y}}
\newcommand\tSigma{\widetilde{\Sigma}}
\newcommand{\grad}{\nabla}
\newcommand\indic{1}
\newcommand{\hauss}{\mathcal{H}}

\newcommand{\pen}{\mathcal{P}}
\newcommand{\dualpen}{\mathcal{Q}}
\newcommand{\dualball}{C}
\newcommand{\Mm}{\mathbb{M}}
\newcommand{\Qq}{\mathbb{Q}}
\newcommand{\Ss}{\mathbb{S}}
\newcommand{\spa}{\mathrm{span}}
\newcommand{\nul}{\mathrm{null}}
\newcommand{\relint}{\mathrm{relint}}

\newcommand{\cG}{\mathcal{G}}
\newcommand{\op}{\mathrm{op}}
\newcommand{\tr}{\mathrm{tr}}

\newcommand{\JLcomment}[1]{\begin{quote}
\color{green} \emph{JL -- #1} \end{quote}}

\newcommand{\hbeta}{\hat{\beta}}
\newcommand{\sgn}{\mathrm{sign}}

\makeatletter
\newcommand{\pushright}[1]{\ifmeasuring@#1\else\omit\hfill$\displaystyle#1$\fi\ignorespaces}
\newcommand{\pushleft}[1]{\ifmeasuring@#1\else\omit$\displaystyle#1$\hfill\fi\ignorespaces}
\makeatother

\endlocaldefs

\begin{document}
\begin{frontmatter}

\title{Inference in Adaptive Regression via the Kac-Rice Formula}
\runtitle{Inference via Kac-Rice}

\begin{aug}
\author{\fnms{Jonathan} \snm{Taylor}\corref{}\ead[label=e1]{jonathan.taylor@stanford.edu}\thanksref{t1},  
\fnms{Joshua} \snm{Loftus}\ead[label=e2]{joftius@stanford.edu},
\fnms{Ryan J.} \snm{Tibshirani}\ead[label=e3]{ryantibs@cmu.edu} 
}
\runauthor{Taylor et al.}

\affiliation{$^1$Stanford University and $^2$Carnegie Mellon
University}

\address{Department of Statistics\\ Stanford University 
\\ Sequoia Hall \\ Stanford, CA 94305, USA \\ \printead{e1} \\
\printead*{e2} } 

\address{Department of Statistics \\
Carnegie Mellon University \\ Baker Hall \\
Pittsburgh, PA 15213, USA \\ \printead{e3}}

\thankstext{t1}{Supported in part by NSF grant DMS 1208857 and
AFOSR grant 113039.}
\end{aug}


\begin{abstract}
We derive an exact p-value for testing a global null
hypothesis in a general adaptive regression setting.  Our
approach uses the Kac-Rice formula \citep[as described in][]{RFG}
applied to the problem of maximizing a Gaussian process. 
The resulting test statistic has a known distribution in
finite samples, assuming Gaussian errors.  We
examine this test statistic in the case of the lasso, group lasso,
principal components, and matrix completion problems. For the lasso
problem, our test relates closely to the recently proposed covariance
test of \citet{covtest}.   Our approach also yields exact selective
inference for the mean parameter at the global maximizer of the process.
\end{abstract}

\begin{keyword}[class=AMS]
\kwd[Primary ]{62M40}
\kwd[; secondary ]{62J05}
\end{keyword}

\begin{keyword}
\kwd{Gaussian processes}
\kwd{convex analysis}
\kwd{regularized regression}
\end{keyword}

\end{frontmatter}

\section{Introduction}
\label{sec:introduction}

In this work, we consider the problem of finding the distribution of 
\begin{equation}
\label{eq:support}
\max_{\eta \in \K} \, \eta^T\epsilon, \qquad 
\epsilon \sim N(0,\Theta), 
\end{equation}
for a convex set $\K \subseteq \real^p$. 
In other words, we study a Gaussian process with a
finite Karhunen-Lo\`eve expansion \citep{RFG}, restricted to a 
convex set in $\real^p$.  

While this is a well-studied topic in the literature of Gaussian
processes, our aim here is to describe an implicit formula for both
the distribution of \eqref{eq:support}, as well as the almost surely
unique maximizer  
\begin{equation}
\label{eq:maximizer}
\eta^* = \argmax_{\eta \in \K} \, \eta^T\epsilon.
\end{equation}
A main point of motivation underlying our work is the
application of such a formula for inference in modern statistical
estimation problems.  We note that a similar (albeit simpler) formula
has proven useful in problems related to sparse regression
\citep{covtest,taylor21}.  Though the general setting considered in 
this paper is ultimately much more broad, we begin by discussing the 
sparse regression case.

\subsection{Example: the lasso}
\label{sec:lasso}

As a preview, consider the $\ell_1$ penalized regression problem,
i.e., the lasso problem \citep{tibshirani:lasso}, of the form
\begin{equation}
\label{eq:lasso}
\hbeta = \argmin_{\beta \in \real^p} \,
\frac{1}{2} \|y-X\beta\|^2_2+ \lambda \|\beta\|_1,
\end{equation}
where $y\in\real^n$, 
$X \in \real^{n\times p}$, and $\lambda \geq 0$. 
Very mild conditions on the predictor matrix $X$ ensure uniqueness of 
the lasso solution \smash{$\hbeta$}, see, e.g.,
\citet{lassounique}. Treating $X$ as fixed, we assume that the outcome
$y$ satisfies  
\begin{equation}
\label{eq:ynormal1}
y \sim N(X\beta_0,\Sigma),
\end{equation}
where $\beta_0 \in \real^p$ is some fixed true (unknown) coefficient
vector, and $\Sigma\in\real^{p\times p}$ is a known covariance matrix.
In the following sections, we derive a formula that enables a test of
a global null hypothesis $H_0$ in a general regularized regression
setting.  Our main result, Theorem \ref{thm:mainresult}, can be
applied to the lasso problem in order to test the null hypothesis 
$H_0 : \beta_0=0$.  This test 
involves the quantity  
\begin{equation*}
\lambda_1 = \|X^Ty\|_{\infty},
\end{equation*}
which can be seen as the first knot (i.e., critical value) in the
lasso solution path over the regularization parameter $\lambda$
\citep{lars}.  Recalling the duality of the $\ell_1$ and $\ell_\infty$
norms, we can rewrite this quantity as 
\begin{equation}
\label{eq:lassolam1}
 \lambda_1= \max_{\eta \in \K} \, \eta^T (X^Ty),
\end{equation}
where $\K=\{ \eta:  \|\eta\|_1 \leq 1\}$, showing that
$\lambda_1$ is of the form \eqref{eq:support}, with $\epsilon=X^Ty$  
(which has mean zero under the null hypothesis). 
Assuming uniqueness of the entries of $X^Ty$, the maximizer $\eta^*$ 
in \eqref{eq:lassolam1} is
\begin{equation*}
\eta^*_j =
\begin{cases}
\text{sign}(X_j^Ty) & |X_j^Ty| = \|X^Ty\|_{\infty} \\ 
0 & \text{otherwise}
\end{cases}, \qquad j=1,\ldots p.
\end{equation*}
Let $j^*$ denote the maximizing index, so that
 $|X_{j^*}^T y| = \|X^T y\|_\infty$, and also $s^*= \sgn(X_{j^*}^Ty)$,   
$\Theta_{jk} = X_j^T\Sigma X_k$. 
To express our test statistic, we define
\begin{align*}
\V^-_{\eta^*} &= 
\max_{\substack{s \in \{-1,1\}, \, k \neq j \\
 1 - s\Theta_{j^*k}/\Theta_{j^*j^*} > 0}}  
\frac{s(X_k - \Theta_{j^*k}/\Theta_{j^*j^*} X_{j^*})^Ty}
{1 - s \Theta_{j^*k}/\Theta_{j^*j^*}}, \\ 
\V^+_{\eta^*} &= \min_{\substack{s \in \{-1,1\}, \, k \neq j \\
 1 - s\Theta_{j^*k}/\Theta_{j^*j^*} < 0}} 
\frac{s(X_k - \Theta_{j^*k}/\Theta_{j^*j^*} X_{j^*})^Ty}
{1 - s \Theta_{j^*k}/\Theta_{j^*j^*}}.
\end{align*}
Then under $H_0: \beta_0=0$, we prove that
\begin{equation}
\label{eq:lassopval1}
\frac{\Phi\big(\V^+_{\eta^*} / \Theta_{j^*j^*}^{1/2}\big) -
\Phi\big(\lambda_1  / \Theta_{j^*j^*}^{1/2}\big) }
{\Phi \big(\V^+_{\eta^*} / \Theta_{j^*j^*}^{1/2}\big) - 
\Phi\big(\V^-_{\eta^*} / \Theta_{j^*j^*}^{1/2}\big) } \sim
\mathrm{Unif}(0,1),
\end{equation}
where $\Phi$ is the standard normal cumulative 
distribution function. 
This formula is somewhat remarkable, in that it is exact---not
asymptotic in $n,p$---and relies only on the assumption of normality
for $y$ in \eqref{eq:ynormal1} (with essentially no real restrictions
on the predictor matrix $X$). As mentioned above, it is a special case
of Theorem \ref{thm:mainresult}, the main result of this paper.   

The above test statistic \eqref{eq:lassopval1}, which we refer to as the {\em Kac-Rice test} for the LASSO, may seem
complicated, but when the predictors are standardized, $\|X_j\|_2=1$
for $j=1,\ldots p$, and the observations are independent with (say)
unit marginal variance, $\Sigma=I$,
then \smash{$\V^-_{\eta^*}$} is equal to the second knot $\lambda_2$
in the lasso path and \smash{$\V^+_{\eta^*}$} is equal to $\infty$.
Therefore \eqref{eq:lassopval1} simplifies to 
\begin{equation}
\label{eq:lassopval2}
\frac{1 - \Phi(\lambda_1)}
{1 - \Phi(\lambda_2)} \sim \mathrm{Unif}(0,1).
\end{equation}
This statistic measures the relative sizes of $\lambda_1$ and 
$\lambda_2$, with values of $\lambda_1 \gg \lambda_2$ being
evidence against the null hypothesis.

\begin{figure}
\begin{center}
\subfigure[Kac-Rice test]{
\label{fig:lassopvala}
\includegraphics[width=0.5\textwidth]{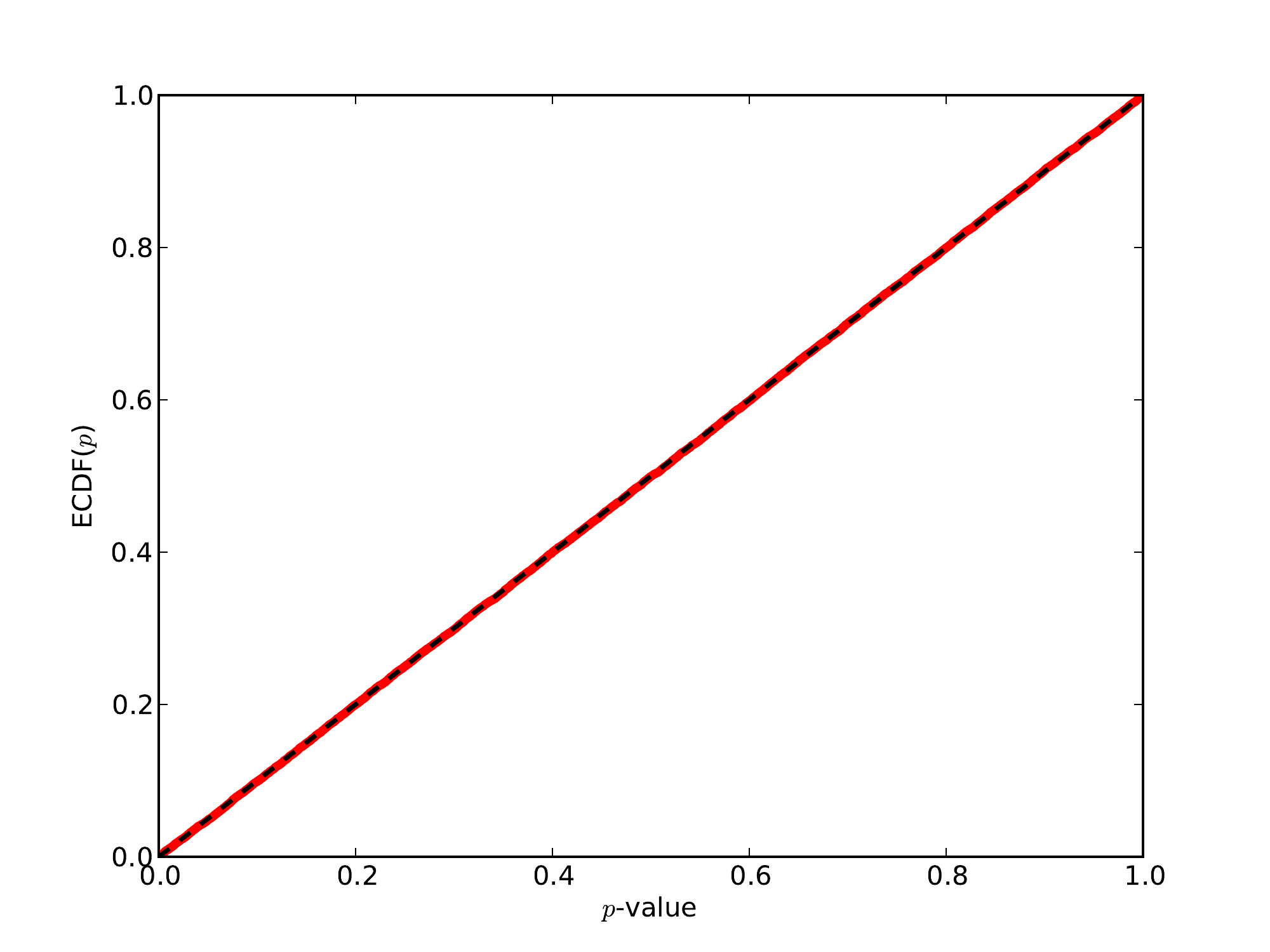}}
\hspace{-15pt}
\subfigure[Covariance test]{
\label{fig:lassopvalb}
\includegraphics[width=0.5\textwidth]{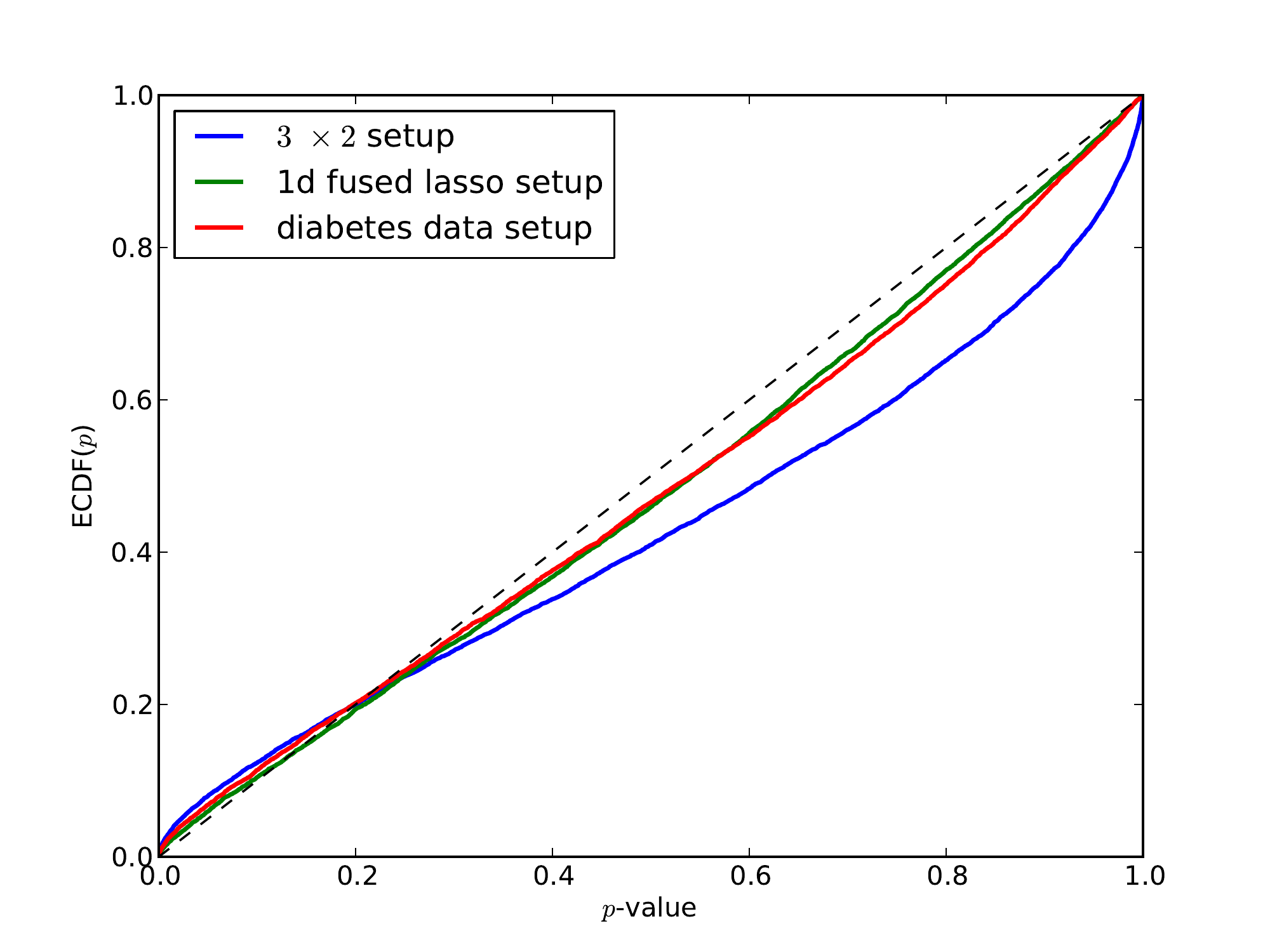}}
\caption{\small \it The left panel shows the empirical distribution
  function of a sample of 20,000 p-values \eqref{eq:lassopval1} coming
  from a variety of different lasso setups.  The agreement with
  uniform here is excellent.  The right panel shows the empirical
  distribution function of a sample of 10,000 covariance test
  p-values, computed using an $\mathrm{Exp}(1)$ approximation,
  using three different lasso setups.  The
  $\mathrm{Exp}(1)$ approximation is generally conservative whereas the 
  Kac-Rice test is exact. }
\end{center}
\end{figure}

Figure \ref{fig:lassopvala} shows the empirical distribution function
of a sample of 20,000 p-values \eqref{eq:lassopval1}, constructed from
lasso problems with a variety
of different predictor matrices, all under the global null model
$\beta_0=0$.  In particular, for each
sample, we drew the predictor matrix $X$ uniformly at random from the
following cases: 
\begin{itemize}
\item small case: $X$ is ${3 \times 2}$, with values (in
  row-major order) equal to $1,2,\ldots 6$;
\item fat case: $X$ is ${100 \times 10,000}$, with columns drawn
  from the compound symmetric Gaussian distribution having correlation
  0.5; 
\item tall case: $X$ is ${10,000 \times 100}$, with columns drawn
  from the compound symmetric Gaussian distribution having correlation 
  0.5;  
\item lower triangular case: $X$ is ${500 \times 500}$, a lower
  triangular matrix of 1s [the lasso problem here is
  effectively a reparametrization of the 1-dimensional fused lasso
  problem \citep{fusedlasso}]; 
\item diabetes data case: $X$ is ${442 \times 10}$, the diabetes
  data set studied in \citet{lars}.
\end{itemize}
As is seen in the plot, the agreement with uniform is very strong.

In their proposed {\it covariance test}, \citet{covtest} show that
under the global null hypothesis $H_0 : \beta_0 = 0$,
\begin{equation}
\label{eq:covtest}
\lambda_1 (\lambda_1 - \lambda_2)
\overset{d}{\rightarrow} \mathrm{Exp}(1) \quad
\text{as} \quad n,p \rightarrow \infty, 
\end{equation}
assuming standardized predictors, $\|X_j\|_2=1$ for $j=1,\ldots p$,
independent errors in \eqref{eq:ynormal1} with unit marginal variance, 
$\Sigma=I$, and a condition to ensure that $\lambda_2$ diverges to
$\infty$ at a sufficient rate.  

In finite samples, using $\mathrm{Exp}(1)$ as an approximation to the
distribution of the covariance test statistic seems generally 
conservative, especially for smaller values of $n$ and $p$. Figure
\ref{fig:lassopvalb} shows the empirical distribution function from
10,000 covariance test p-values, in three of the above scenarios.
[The predictors were standardized before applying the covariance test,
in all three cases; this is not necessary, as the covariance test can
be adapted to the more general case of unstandardized predictors, but
was done for simplicity, to match the form of the test as written in 
\eqref{eq:covtest}.]  Even though the idea behind the covariance test
can be conceivably extended to other regularized regression problems
(outside of the lasso setting), the $\mathrm{Exp}(1)$ approximation 
to its distribution is generally inappropriate, as we will 
see in later examples. Our test, however, naturally extends to
general regularization settings, allowing us to attack
problems with more complex penalties such as the group lasso and
nuclear norm penalties. 

The rest of this paper is organized as follows.
In Section \ref{sec:generalpen}, we describe the general framework for   
regularized regression problems that we consider, and a corresponding
global null hypothesis of interest; we also state our main result,
Theorem \ref{thm:mainresult}, which gives an exact p-value for this
null hypothesis. The next two sections are then dedicated to proving 
Theorem \ref{thm:mainresult}.  Section \ref{sec:character}
characterizes the global maximizer \eqref{eq:maximizer} in terms of
the related Gaussian process and its gradient.  Section
\ref{sec:kacrice} applies the Kac-Rice formula to derive the joint
distribution of the maximum value of the process \eqref{eq:support}
and its maximizer \eqref{eq:maximizer}, which is ultimately used to
derive the (uniform) distribution of our proposed test.  In Section
\ref{sec:examples} we broadly consider practicalities associated with  
constructing our test statistic, revisit the lasso problem, and
examine the group lasso, principal components, and matrix completion
problems as well.  Section \ref{sec:linfrac} discusses the details
of the computation of the quantities $\V^-_{\eta^*}$ and
$\V^+_{\eta^*}$ needed for our main result.  We empirically
investigate the null distribution of our test statistic under
non-Gaussian errors in Section \ref{sec:nongaussian}, and end with a
short discussion in Section \ref{sec:discussion}.

\section{General regularized regression problems}
\label{sec:generalpen}

We examine a class of regularized least squares problems of the form 
\begin{equation}
\label{eq:genprob}
\hbeta \in \argmin_{\beta \in \real^p} \,
\frac{1}{2} \|y-X\beta\|^2_2 + \lambda\cdot\pen(\beta),
\end{equation}
with outcome $y\in\real^n$, predictor matrix $X\in\real^{n\times p}$, 
and regularization parameter $\lambda \geq 0$.  We assume that the
penalty function $\pen$ satisfies
\begin{equation}
\label{eq:seminorm}
\pen(\beta) = \max_{u \in \dualball}\, u^T\beta,
\end{equation}
where $\dualball \subseteq \real^p$ is a convex set,
i.e., $\pen$ is the support function of $\dualball$.  This
serves as a very general penalty, as we can represent any seminorm
(and hence any norm) in this form with the proper choice of set
$\dualball$. In this work, we will use the abuse of notation
of calling \eqref{eq:seminorm} a semi-norm (that is, we do not
require symmetry of semi-norms).  We note that the solution 
\smash{$\hbeta=\hbeta_\lambda$} above is
not necessarily unique (depending on the matrix $X$ and set
$\dualball$) and the element notation used in  
\eqref{eq:genprob} reflects this.  A standard calculation,
which we give in Appendix \ref{app:uniqueness}, shows that 
the fitted value \smash{$X\hbeta$} is always unique, and hence so is
\smash{$\pen(\hbeta)$}.  

Now define
\begin{equation*}
\lambda_1 = \min\{\lambda \geq 0 : 
\pen(\hbeta_\lambda) =0 \}.
\end{equation*}
This is the smallest value of $\lambda$ for which the penalty term in 
\eqref{eq:genprob} is zero; any smaller value of the tuning
parameter returns a non-trivial solution, according to the penalty. 
A straightforward calculation involving subgradients, which we 
give in Appendix \ref{app:genlam1}, shows that
\begin{equation}
\label{eq:genlam1}
\lambda_1 = \dualpen \big(X^T(I-P_{X\dualball^\perp}) y\big),
\end{equation}
where $\dualpen$ is the dual seminorm of $\pen$, i.e.,
$\dualpen(\beta) = \max_{v \in \dualball^{\circ}} v^T \beta$, the
support function of the polar set $\dualball^\circ$ of
$\dualball$. This set can be defined as 
$\dualball^\circ = \{v : v^Tu \leq 1 \;\,\text{for all}\;\, u \in
\dualball\}$, or equivalently, 
\begin{equation*}
\dualball^{\circ} = \{v: \pen(v) \leq 1\},
\end{equation*}
the unit ball in $\pen$.  Furthermore, in
\eqref{eq:genlam1}, we use $P_{X\dualball^\perp}$ to denote the
projection operator onto the linear subspace    
\begin{equation*}
X\dualball^\perp = X \{v : v \perp \dualball\} \subseteq \real^n.
\end{equation*}
We recall that for the lasso problem \eqref{eq:lasso}, the penalty
function is $\pen(\beta)=\|\beta\|_1$, so
$\dualpen(\beta)=\|\beta\|_\infty$; also $\dualball = \{ u :
\|u\|_\infty \leq 1\}$, which means that $P_{X\dualball^\perp}=0$,
and hence $\lambda_1 = \|X^T y\|_\infty$, as claimed in Section 
\ref{sec:lasso}.  

Having just given an example of a seminorm in which 
$\dualball$ is of full dimension $p$, so that 
$\dualball^\perp=\{0\}$,
we now consider one in which $\dualball$ has dimension less 
than $p$, so that $\dualball^\perp$ is nontrivial.
In a generalized lasso problem \citep{genlasso1}, the penalty is
$\pen(\beta)=\|D\beta\|_1$ for some chosen penalty matrix 
$D \in \real^{m\times n}$. 
In this case, it can be shown that the dual seminorm is
$\dualpen(\beta) = \min_{D^T z = \beta} \|z\|_\infty$. 
Hence $\dualball=\{u : \min_{D^T z = u} \|z\|_\infty \leq 1\}$, 
and $\dualball^\perp=\nul(D)$, the null space of $D$. In many interesting 
cases, this null space is nontrivial; e.g., if $D$ is the fused lasso 
penalty matrix, then its null space is spanned by the vector of all
1s.
In fact, the usual form of the LASSO \cite{tibshirani:lasso} is 
a seminorm:
$$
\minimize_{\beta \in \real^p, \gamma \in \real} \frac{1}{2} \|y-(\gamma 1+X\beta)\|^2_2 + \lambda \|\beta\|_1, \qquad y \sim N(X\beta, \sigma^2 I).
$$
Of course, the above problem can be solved by solving
$$
\minimize_{\beta \in \real^p} \frac{1}{2} \|z-PX\beta\|^2_2 + \lambda \|\beta\|_1, \qquad z = Py \sim N(X\beta, \sigma^2 P), P = I - \frac{1}{n} 11^T.
$$
Both of these problems fit into the framework \eqref{eq:genprob}.

\subsection{A null hypothesis}
\label{sec:nullhyp}

As in the lasso case, we assume that $y$ is
generated from the normal model
\begin{equation}
\label{eq:ynormal2}
y \sim N(X\beta_0,\Sigma),
\end{equation}
with $X$ considered fixed.  We are interested in the distribution
of \eqref{eq:support} in order to test the following hypothesis:   
\begin{equation}
\label{eq:null1}
H_0: \pen(\beta_0)=0.
\end{equation}
This can be seen a global null hypothesis, a test of 
whether the true underlying coefficient
vector $\beta_0$ has a trivial structure, according to the
designated penalty function $\pen$. 

Assuming that the set $\dualball$
contains 0 in its relative interior, we have $\pen(\beta)=0
\Longleftrightarrow P_\dualball \beta = 0$, where $P_\dualball$ denotes
the projection matrix onto $\spa(\dualball)$.  Therefore we can
rewrite the null hypothesis \eqref{eq:null1} in a more transparent
form, as
\begin{equation}
\label{eq:null2}
H_0 : P_\dualball \beta_0 = 0.  
\end{equation}
Again, using the lasso problem \eqref{eq:lasso} as a reference
point, we have $\spa(\dualball)=\real^p$ for this problem, so the
above null hypothesis reduces to $H_0: \beta_0=0$, as in 
Section \ref{sec:lasso}.  In general, the null hypothesis \eqref{eq:null2}
tests $\beta_0 \in \dualball^\perp$, the orthocomplement of 
$\dualball$. 

Recalling that
\begin{equation*}
\lambda_1 = \max_{v  \in \dualball^\circ} \, 
v^T X^T(I-P_{X\dualball^\perp}) y,
\end{equation*}
one can check that, under $H_0$, the quantity
$\lambda_1$ is precisely of the form \eqref{eq:support}, with
$\K=\dualball^\circ$, $\epsilon = X^T  (I-P_{X\dualball^\perp})y$, and
$\Theta=\Cov(\epsilon)= X^T(I-P_{X\dualball^\perp})\Sigma 
(I-P_{X\dualball^\perp}) X$ [as
$\Ee(\epsilon)=X^T(I-P_{X\dualball^\perp})X\beta_0=0$ when $\beta_0
\in \dualball^\perp$].



\subsection{Statement of main result and outline of our approach}
\label{sec:mainresult}

We now state our main result.  

\begin{theorem}[Main result: Kac-Rice test]
\label{thm:mainresult}
Consider the general regularized regression problem in
\eqref{eq:genprob}, with
$\pen(\beta)=\max_{u\in\dualball} u^T \beta$ for
a closed, convex set $\dualball \subseteq \real^p$ containing 0
in its relative interior.  Denote  
$\K = \dualball^\circ = \{v:\pen(v) \leq 1\}$, the polar set of
$\dualball$, and assume that $\K$ can be stratified into
pieces of different dimensions, i.e., 
\begin{equation}
\label{eq:strata}
\K = \bigcup_{j=0}^p \partial_j \K,
\end{equation}
where $\partial_0\K, \ldots \partial_p \K$ are smooth disjoint
manifolds of dimensions $0,\ldots p$, respectively. 
Assume also assume that the process
\begin{equation}
\label{eq:f}
f_\eta = \eta^T X^T (I-P_{X\dualball^\perp}) y, \qquad \eta \in  
\K, 
\end{equation}
is Morse for almost every $y \in \real^n$.  Finally, assume that
$y \in \real^n$ is drawn from the normal distribution in
\eqref{eq:ynormal2}.   

Now, consider testing the null hypothesis 
$H_0 : \pen(\beta_0)=0$ [equivalently, $H_0: P_\dualball \beta_0=0$,
since we have assumed that $0 \in \relint(\dualball)$].  Define 
$\Lambda_\eta = G^{-1}_\eta H_\eta$ for $G_\eta,H_\eta$ as in
\eqref{eq:g}, \eqref{eq:h},  $\V^-_\eta,\V^+_\eta$ as in
\eqref{eq:vplus}, \eqref{eq:vminus}, and $\sigma^2_\eta$ as in
\eqref{eq:sigma}.  Finally, let $\eta^*$ denote the almost sure
unique maximizer of the process $f_\eta$ over $\K$, 
\begin{equation*}
\eta^* = \argmax_{\eta \in \K} \, f_\eta,
\end{equation*} 
and let $\lambda_1=f_{\eta^*}$ denote the first knot in the solution
path of problem \eqref{eq:genprob}.  Then under $H_0$,
\begin{equation}
\label{eq:sunif1}
\frac{\displaystyle \int_{\lambda_1}^{\V^+_{\eta^*}}
\det(\Lambda_{\eta^*}+z I) \phi_{\sigma^2_{\eta^*}} (z) \, dz}
{\displaystyle \int_{\V^-_{\eta^*}}^{\V^+_{\eta^*}}
\det(\Lambda_{\eta^*}+z I) \phi_{\sigma^2_{\eta^*}} (z) \, dz}
\sim \mathrm{Unif}(0,1),
\end{equation}
where $\phi_{\sigma^2}$ denotes the density function of a normal
random variable with mean 0 and variance $\sigma^2$. 
\end{theorem}

The quantity \eqref{eq:sunif1} is the {\em Kac-Rice pivot}
evaluated at $\mu=0$. Lemma \ref{lem:mainresult} shows it is
a pivotal quantity for the mean $\mu$ near $\eta^*$, derived via the
Kac-Rice formula. 
Here we give a rough explanation of the result in \eqref{eq:sunif1},
and the approach we take to prove it in Sections \ref{sec:character}
and \ref{sec:kacrice}.  The next section, Section \ref{sec:assump},
discusses the assumptions behind Theorem \ref{thm:mainresult}; in
summary, 
the assumption that $\K$ separates as in \eqref{eq:strata} allows us
to apply the Kac-Rice formula to each of its strata, and the Morse
assumption on the process $f_\eta$ in \eqref{eq:f} ensures the
uniqueness of its maximizer $\eta^*$.  These are very 
weak assumptions, especially considering the strength of the exact, 
non-asymptotic conclusion in \eqref{eq:sunif1}.   

Our general approach
is based on finding an implicit formula for $\Pp(\lambda_1 > t)$
under the null hypothesis $H_0$, where $\lambda_1$ is the first
knot in the solution path of problem
\eqref{eq:genprob} and can be written as
\begin{equation*}
\lambda_1 = \max_{\eta \in \K} \, f_\eta,
\end{equation*}
where $f_\eta= \eta^T X^T (I-P_{X\dualball^\perp}) y$, the process in
\eqref{eq:f}.  Our representation for the tail probability of
$\lambda_1$ has the form
\begin{equation}
\label{eq:implicit}
\Pp(\lambda_1 > t) = \Ee\big(\Qq(\indic_{(t,\infty)})\big).
\end{equation}
Here $\Qq=\Qq_{\eta^*}$ is a random distribution function,
$\indic_{(t,\infty)}$ is the indicator function for the
interval $(t,\infty$), and $\eta^*$ is a maximizer of the process 
$f_\eta$.  This maximizer $\eta^*$, almost surely unique by 
the Morse assumption, satisfies 
\begin{equation*}
\eta^* \in 
\partial \dualpen \big( X^T (I-P_{X\dualball^\perp})) y\big)
\subseteq \K,
\end{equation*}
with $\partial \dualpen$ the subdifferential of the seminorm 
$\dualpen$.  Under the assumption that $\K =
\cup_{j=0}^p \partial_j \K$, the main tool we invoke is the Kac-Rice
formula \citep{RFG}, which essentially enables us to compute the
expected number of global maximizers occuring in each stratum
$\partial_j \K$.

\begin{remark}
Note that for almost every realization, under the Morse assumption, there is generically only
one maximizer overall and hence the number of them is either 0 or 1.
We use the term "number of global maximizers" when applying the Kac-Rice formula
as it applies to counting different types of points. In our applications of it, however, there is only
ever 0 or 1 such points. Similar arguments were used to establish
the accuracy of the expected Euler characteristic approximation for the distribution 
of the global maximum of a smooth Gaussian process in \cite{validity}.
\end{remark}

This leads to the distribution
of $\lambda_1$, in Theorem \ref{thm:dist}, as well as the 
representation in \eqref{eq:implicit}, with $\Qq$ given in an explicit
form. Unfortunately, computing tail
probabilities $\Pp(\lambda_1 > t)$ of this distribution involve
evaluating some complicated integrals over $\K$ that depend on
$X,\Sigma$, and hence the quantity $\lambda_1$ as a test statistic
does not easily lend itself to the computation of p-values.  We
therefore turn to the survival function $\Ss_{\eta^*}$ associated with
the measure $\Qq_{\eta^*}$, and our main result is that, when
carefully evaluated, this (random) survival function can be used to
derive a test of $H_0$, as expressed in \eqref{eq:sunif1} in Theorem
\ref{thm:mainresult} above.   




\subsection{Discussion of assumptions}
\label{sec:assump}

In terms of the assumptions of Theorem \ref{thm:mainresult}, we
require that $\dualball$ contains 0 in its relative interior so that
we can write the null hypothesis in the equivalent form $H_0 : P_C
\beta_0 = 0$, which makes the process $f_\eta$ in \eqref{eq:f} have
mean zero under $H_0$.  We additionally assume that $\dualball$ is
closed in order to guarantee that $f_\eta$ has a 
well-defined (finite) maximum over $\eta \in \K =
\dualball^\circ$.  See Appendix \ref{app:welldefined}.

Apart from these rather minor assumptions on $\dualball$, the main   
requirements of the theorem are: the polar set
$\dualball^\circ=\K$ can be stratified as in 
\eqref{eq:strata}, the process $f_\eta$ in \eqref{eq:f} is Morse, and
$y$ follows the normal distribution in \eqref{eq:ynormal2}. Overall,
these are quite weak assumptions.  The first assumption, on $\K$
separating as in \eqref{eq:strata}, eventually permits us to apply to
the Kac-Rice formula to each stratum $\partial_j \K$.  We remark  
that many convex (and non-convex) sets possess   
such a decomposition; see \citet{RFG}.  In particular, we note that
such an assumption does not limit our consideration to polyhedral
$\K$: a set can be stratifiable but still have a boundary with
curvature (e.g., as in $\K$ for the group lasso and nuclear norm
penalties).   

Further, the property of being a Morse function is truly generic;  
again, see \citet{RFG} for a discussion of Morse functions on
stratified spaces.  If $f_\eta$ is Morse for almost every $y$, then
its maximizers are almost surely isolated, and the convexity of $\K$
then implies that $f_\eta$ has an almost surely unique
maximizer $\eta^*$. From the form of our particular process $f_\eta$
in \eqref{eq:f}, the assumption that 
$f_\eta$ is Morse can be seen as a restriction on the predictor matrix
$X$ (or more generally, how $X$ interacts with the set $\dualball$).
For most problems, this only rules out trivial choices of $X$.  In the
lasso case, for example, recall that 
$f_\eta = (X\eta)^T y$ and $\K$ is equal to the unit $\ell_1$ ball, so   
\smash{$f_\eta^* = \|X^T y\|_\infty$}, and the
Morse property requires \smash{$|X_j^T y|$}, 
$j=1,\ldots p$ to be unique for almost every $y \in \real^p$.  This
can be ensured by taking $X$ with columns in general position [a weak  
condition that also ensures uniqueness of the lasso solution; see
\citet{lassounique}].  

Lastly, the assumption of normally distributed
errors in the regression model \eqref{eq:ynormal2} is important for 
the work that follows in Sections \ref{sec:character} and
\ref{sec:kacrice}, which is based on Gaussian process theory.  Note
that we assume a known covariance matrix $\Sigma$, but we allow for a 
dependence between the errors (i.e., $\Sigma$ need not be diagonal).   
Empirically, the (uniform) distribution of our test statistic under
the null hypothesis appears
to quite robust against non-normal errors in many cases of interest;
we present such simulation results in Section \ref{sec:nongaussian}.

\subsection{Notation}
\label{sec:notation}

Rewrite the process $f_\eta$ in
\eqref{eq:f} as   
\begin{equation*}
f_\eta = 
\eta^T X^T (I-P_{X\dualball^\perp}) (I-P_{X\dualball^\perp}) y = 
\eta^T \tX^T \ty, \qquad \eta \in \K,
\end{equation*}
where \smash{$\tX = (I-P_{X\dualball^\perp}) X$}
and \smash{$\ty = (I-P_{X\dualball^\perp}) y$}.  The
distribution of \smash{$\ty$} is hence 
\smash{$\ty \sim N(\tX \beta_0,\tSigma)$}, where 
\smash{$\tSigma = 
(I-P_{X\dualball^\perp}) \Sigma (I-P_{X\dualball^\perp})$}.  
Furthermore, under the null hypothesis 
$H_0 : P_\dualball \beta_0 = 0$, we have \smash{$\ty \sim
N(0,\tSigma)$}.  For convenience, in Sections \ref{sec:character}
and \ref{sec:kacrice}, we will
drop the tilde notation, and write \smash{$\ty,\tX,\tSigma$} as simply 
$y,X,\Sigma$, respectively.  To be perfectly explicit, this means that 
we will write the process $f_\eta$ in \eqref{eq:f} as 
\begin{equation*}
f_\eta = \eta^T X^T y, \qquad \eta \in \K,
\end{equation*}
where $y \sim N(X\beta_0,\Sigma)$, and the null hypothesis is
$H_0 : y \sim N(0,\Sigma)$.
Notice that when $\spa(\dualball)=\real^p$, we have exactly
\smash{$\ty=y,\, \tX=X,\, \tSigma=\Sigma$}, since
$P_{X\dualball^\perp}=0$.  However, we reiterate that replacing 
\smash{$\ty,\tX,\tSigma$} by $y,X,\Sigma$ in Sections
\ref{sec:character} and \ref{sec:kacrice}
is done purely for notational convenience, and the reader should bear
in mind that the arguments themselves do not portray any loss of 
generality. 

We will write $\Ee_0$ to emphasize that an expectation is taken under
the null distribution $H_0 : y \sim N(0,\Sigma)$.

\section{Characterization of the global maximizer}
\label{sec:character}

Near any point $\eta \in \K$, the set $\K$ is well-approximated by
the support cone $S_{\eta}\K$, which is defined as the polar cone
of the normal cone $N_{\eta}\K$.   
The support cone $S_{\eta}\K$ contains a largest linear subspace---we
will refer to this $T_{\eta}\K$, the {\it tangent space} to $\K$ at
$\eta$. The tangent space plays an important role in what follows.   

\begin{remark}
\label{rem:convexity}
Until Section 6, the convexity of the parameter space
$\K$ is not necessary; we only need local convexity as described in
\citet{RFG}, i.e., we only need to assume that the support cone of
$\K$ is locally convex everywhere. This is essentially the 
same as positive reach \citep{federer}. To be clear,
while convexity is used in connecting the
Kac-Rice test to regularized regression problems in \eqref{eq:genlam1} (i.e.
it establishes an equality between left and right hand sides), the right hand side
is well-defined even if the set $C^{\circ}$ is not convex. That is, the $\K$
 in \eqref{eq:support} need not be convex. In fact, the issue of convexity
 is only important for computational purposes, not theoretical purposes. In this sense,
 this work provides an exact conditional test based on the global maximizer of a
 smooth Gaussian field on a fairly arbitrary set. This is an advance in the theory of
 smooth Gaussian fields as developed in \cite{RFG} and will be investigated
 in future work.
\end{remark}

We study the process $f_\eta$ in \eqref{eq:f}, which we now
write as $f_\eta = \eta^T X^T y$ over $\eta \in \K$, where $y \sim
N(X\beta_0,\Sigma)$, and with the null hypothesis $H_0 : y \sim 
N(0,\Sigma)$ (see our notational reduction in Section
\ref{sec:notation}).  We proceed as in
Chapter 14 of \citet{RFG}, with an important difference being that 
here the process $f_\eta$ does not have constant variance.  Aside from
the statistical implications of this work that have to do with
hypothesis testing, another goal of this 
paper is to derive analogues of the results in \citet{validity,RFG}
for Gaussian processes with nonconstant  
variance.  For each $\eta \in \K$, we define a modified process
\begin{equation*}
\tf^{\eta}_z = f_z - z^T\alpha_{\eta,X,\Sigma}(\nabla
f_{|T_{\eta}\K}), \qquad z \in \K,
\end{equation*}
where $\alpha_{\eta,X,\Sigma}(\nabla f_{|T_{\eta}\K})$ is the vector
that, under $H_0 : y \sim N(0,\Sigma)$, computes the expectation of
$f_z$ given $\nabla f_{|T_{\eta}\K}$, the gradient restricted to
$T_{\eta}\K$, i.e., 
\begin{equation*}
z^T\alpha_{\eta,X,\Sigma}(\nabla f_{|T_{\eta}\K}) = 
\Ee_0(f_z \big| \nabla f_{|T_{\eta}\K}).
\end{equation*} 
To check that such a representation is possible, suppose that the
tangent space $T_\eta \K$ is $j$-dimensional, and let $V_\eta \in
\real^{p\times j}$ be a 
matrix whose columns form an orthonormal basis for $T_\eta\K$.  Then   
\smash{$\nabla f_{|T_\eta\K} = V_\eta V_\eta^T X^T y$}, and a
simple calculation using the properties of conditional expectations for
jointly Gaussian random variables shows that 
\begin{equation*}
\Ee_0(f_z\big|  \nabla f_{|T_{\eta}\K}) = z^T X^T P_{\eta,X,\Sigma} y,
\end{equation*}
where
\begin{equation}
\label{eq:projection}
P_{\eta,X,\Sigma} = \Sigma X V_\eta(V_\eta^T X^T \Sigma
X V_\eta)^\dagger V_\eta^T X^T,
\end{equation}
the projection 
onto $XV_\eta$ with respect to $\Sigma$ (and \smash{$A^\dagger$}
denoting the Moore-Penrose pseudoinverse of a matrix $A$). 
Hence, we gather that
\begin{equation*}
\alpha_{\eta,X,\Sigma}(\nabla  f_{|T_{\eta}\K}) = X^T
P_{\eta,X,\Sigma} y,
\end{equation*}
and our modified process has the form 
\begin{equation}
\label{eq:tildef}
\tf^{\eta}_z = f_z - z^TX^TP_{\eta,X,\Sigma} y 
= (Xz)^T(I - P_{\eta, X,\Sigma}) y.
\end{equation}

The key observation, as in \citet{validity} and \citet{RFG}, is that
if $\eta$ is a critical point, i.e., $\nabla f_{|T_\eta\K}=0$, then  
\begin{equation}
\label{eq:keyobs}
\tf^{\eta}_z=f_z \qquad \text{for all}\;\, z \in \K.
\end{equation}
Similar to our construction of $\alpha_{\eta,X,\Sigma}(\nabla
f_{|T_\eta\K})$, we define $C_{X,\Sigma}(\eta)$ such that 
\begin{align}
\nonumber
\Ee_0(\tf^{\eta}_z \big| \tf^{\eta}_{\eta})
&=
\frac{z^TX^T(I-P_{\eta,X,\Sigma})\Sigma(I-P_{\eta,X,\Sigma}^T)X\eta}
{\eta^TX^T(I-P_{\eta,X,\Sigma}) \Sigma(I - P_{\eta,X,\Sigma}^T)X\eta}  
\cdot \tf^{\eta}_{\eta} \\ 
\label{eq:c}
& = z^T C_{X,\Sigma}(\eta) \cdot \tf^{\eta}_{\eta}.
\end{align}
and after making three subsequent definitions, 
\begin{align}
\label{eq:vminus}
\V^-_{\eta} &= \max_{z \in \K:\, z^T C_{X,\Sigma}(\eta) < 1} \,
\frac{\tf^{\eta}_z - z^T C_{X,\Sigma}(\eta) \cdot \tf^{\eta}_{\eta}}
{1 - z^T C_{X,\Sigma}(\eta)}, \\
\label{eq:vplus}
\V^+_{\eta} &= \min_{z \in \K:\, z^T C_{X,\Sigma}(\eta) > 1} \,
\frac{\tf^{\eta}_z - z^T C_{X,\Sigma}(\eta) \cdot \tf^{\eta}_{\eta} }
{1 - z^T C_{X,\Sigma}(\eta)}, \\ 
\label{eq:vzero}
\V^0_{\eta} &=  \max_{z \in \K:\, z^T C_{X,\Sigma}(\eta) = 1} \,
\tf^{\eta}_z - z^T C_{X,\Sigma}(\eta) \cdot \tf^{\eta}_{\eta}, 
\end{align}
we are ready to state our characterization of the global maximizer
$\eta$. 

\begin{lemma}
\label{lem:globalmax}
A point $\eta \in \K$ maximizes $f_\eta$ over a convex
set $\K$ if and only if the following conditions hold: 
\begin{equation}
\label{eq:maxcond}
\grad f_{|T_{\eta}\K} = 0, \quad
\tf^{\eta}_{\eta} \geq \V^-_{\eta}, \quad
\tf^{\eta}_{\eta} \leq \V^+_{\eta}, \quad \text{and} \quad
\V^0_{\eta} \leq 0.
\end{equation}
The same equivalence holds true even when $\K$ is only locally
convex. 
\end{lemma}

\begin{proof}
In the forward direction ($\Rightarrow$), note that 
\smash{$\nabla f_{|T_\eta\K}=0$} implies that we can replace 
\smash{$\tf^\eta_z$} by $f_z$ (and \smash{$\tf^\eta_\eta$} by
$f_\eta$) in the definitions \eqref{eq:vplus}, \eqref{eq:vminus},  
\eqref{eq:vzero}, by the key observation \eqref{eq:keyobs}. 
As each $z\in\K$ is covered by one of the cases 
$C_{X,\Sigma}(\eta)<1$, $C_{X,\Sigma}(\eta)>1$,
$C_{X,\Sigma}(\eta)=1$, we conclude that 
\begin{equation*}
f_{\eta} \geq f_z \qquad \text{for all}\;\, z \in \K,
\end{equation*}
i.e., the point $\eta$ is a global maximizer. 

As for the reverse direction ($\Leftarrow$), when $\eta$ is the global 
maximizer of $f_\eta$ over $\K$, the first condition $\nabla
f_{|T_\eta\K}=0$ is clearly true (provided that $\K$ is convex or
locally convex), and the other three conditions follow from simple
manipulations of the inequalities   
\begin{equation*}
f_{\eta} \geq f_z \qquad \text{for all}\;\, z \in \K.
\end{equation*}
\end{proof}

\begin{remark}
The above lemma does not assume that $\K$ decomposes
into strata, or that $f_\eta$ is Morse for almost all $w$, or 
that $y \sim N(X\beta_0,\Sigma)$.  It only assumes that $\K$ is convex 
or locally convex, and its conclusion is completely deterministic, 
depending only on the process $f_\eta$ via its covariance function
under the null, i.e., via the terms $P_{\eta,X,\Sigma}$ and
$C_{X,\Sigma}(\eta)$.   

We note that, under the assumption that $f_\eta$ is Morse over $\K$
for almost every $y\in\real^p$, and $\K$ is convex, Lemma
\ref{lem:globalmax} gives necessary and sufficient conditions for a
point $\eta \in \K$ to be the almost sure unique global maximizer. 
Hence, for convex $\K$, the conditions in \eqref{eq:maxcond} are
equivalent to the usual subgradient conditions for optimality, which
may be written as 
\begin{equation*}
\nabla f_{\eta} \in N_{\eta}\K \quad\iff\quad 
\nabla f_{|T_{\eta}\K}=0, \; 
\nabla f_{|(T_\eta\K)^\perp} \in N_{\eta}\K,
\end{equation*}
where $N_\eta\K$ is the normal cone to $\K$ at $\eta$ and
$\nabla f_{|(T_\eta\K)^\perp}$ is the gradient restricted to the 
orthogonal complement of the tangent space $T_\eta\K$. 
\end{remark}

Recalling that $f_\eta$ is a Gaussian process, a helpful 
independence relationship unfolds.  

\begin{lemma}
\label{lem:indep}
With $y \sim N(X\beta_0,\Sigma)$, for each fixed 
$\eta \in \K$, the triplet 
\smash{$ (\V^-_{\eta}, \V^+_{\eta}, \V^0_{\eta})$}
is independent of \smash{$\tf^{\eta}_{\eta}$}.
\end{lemma}

\begin{proof}
This is a basic property of conditional expectation for jointly
Gaussian variables, i.e., it is easily verified that 
\smash{$\Cov(\tf^\eta_z - z^T C_{X,\Sigma}(\eta), \tf_\eta^\eta) = 0$} 
for all $z$.
\end{proof}




\section{Kac-Rice formulae for the global maximizer and its value} 
\label{sec:kacrice}

The characterization of the global maximizer from the last
section, along with the Kac-Rice formula \citep{RFG}, allow  
us to express the joint distribution of
\begin{equation*}
\eta^* = \argmax_{\eta \in \K}\,  f_{\eta} \qquad\text{and}\qquad
f_{\eta^*} = \max_{\eta\in\K}\, f_\eta.
\end{equation*}

\begin{theorem}[Joint distribution of $(\eta^*,f_{\eta^*})$]
\label{thm:dist}
Writing \smash{$\K = \cup_{j=0}^p \partial_j\K$} for a
stratification of $\K$, 
for open sets $A \subseteq \real^p, O \subseteq \real$,    
\begin{multline}
\label{eq:dist}
\Pp(\eta^* \in A, \, f_{\eta^*} \in O) = 
\sum_{j=0}^p \int_{\partial_j\K \cap A} \Ee 
\bigg(\det(-\nabla^2 f_{|T_{\eta}\K}) \; \cdot \\
\indic_{\{\V^-_{\eta} \leq \tf^{\eta}_{\eta} \leq \V^+_{\eta}, 
\V^0_{\eta} \leq 0, \tf_{\eta}^{\eta} \in O\}} \,\bigg|\, 
\nabla f_{|T_{\eta}\K} =0 \bigg)  
\psi_{\nabla f_{|T_{\eta}\K}}(0) \, \hauss_j(d\eta),
\end{multline}
where:
\begin{itemize}
\item $\psi_{\nabla f_{|T_{\eta}\K}}$ is the density of the gradient
in some basis for the tangent space $T_{\eta}\K$, orthonormal with
respect to the standard inner product on $T_{\eta}\K$, i.e., the
standard Euclidean Riemannian metric on $\real^p$;  
\item the measure $\hauss_j$ is the Hausdorff measure induced by the
above Riemannian metric on each $\partial_j\K$; 
\item the Hessian $\nabla^2 f_{|T_{\eta}\K}$ is evaluated in this
orthonormal basis and, for $j=0$, we take as convention the
determinant of a $0 \times 0$ matrix to be 1 [in \citet{RFG}, this 
was denoted by $\nabla^2 f_{|\partial_j\K,\eta}$, to emphasize that it 
is the Hessian of the restriction of $f$ to $\partial_j\K$].
\end{itemize}
\end{theorem}

\begin{proof}
This is the Kac-Rice formula, or the ``meta-theorem'' of
Chapter 10 of \citet{RFG} \citep[see also][]{azaisW,brillinger},
applied to the problem of counting the number of global maximizers in
some set $A \subseteq \real^p$ having value in $O \subseteq
\real$. That is,     
\begin{align*}
\Pp(\eta^* \in A, \, f_{\eta^*} \in O)
&= \Ee\Big(\# \Big\{\eta \in \K \cap A: \nabla
    f_{|T_{\eta}\K}=0, 
\\ & \hspace{100pt}
 \V^-_{\eta} \leq \tf^{\eta}_{\eta} \leq
    \V^+_{\eta}, \, \V^0_{\eta} \leq 0, \, f_{\eta} \in O \Big\} \Big)
\\  &= \Ee\Big(\# \Big\{\eta \in \K \cap A: \nabla 
    f_{|T_{\eta}\K}=0, 
\\ & \hspace{100pt}
\V^-_{\eta} \leq \tf^{\eta}_{\eta} \leq
    \V^+_{\eta}, \, \V^0_{\eta} \leq 0, \, \tf^{\eta}_{\eta} \in O \Big\} 
\Big),
\end{align*}
where the second equality follows from \eqref{eq:keyobs}.  Breaking
down $\K$ into its separate strata, and then using the Kac-Rice
formula, we obtain the result in \eqref{eq:dist}.
\end{proof}

\begin{remark}
As before, the conclusion of Theorem \ref{thm:dist} does not actually 
depend on the convexity of $\K$.  When $\K$ is only locally 
convex, the Kac-Rice formula [i.e., the right-hand side in
\eqref{eq:dist}] counts the expected
total number of global maximizers of $f_\eta$ lying in some set $A \subseteq
\real^p$, with the achieved maximum value in $O \subseteq
\real$. For convex $\K$, our Morse condition on $f_\eta$ implies an 
almost surely unique maximizer, and hence the notation
\smash{$\Pp(\eta^* \in  A,\,f_{\eta^*} \in O)$} on the left-hand side
of \eqref{eq:dist} makes sense as written.  For locally convex $\K$,
one simply needs to interpret the left-hand side as  
\begin{equation*}
\Pp\Big(\eta^* \in A \;\,\text{for some}\;\,
f_{\eta^*} = \max_{\eta\in\K}\, f_\eta,\;\, \max_{\eta\in\K} \, f_\eta
\in O\Big).
\end{equation*}
\end{remark}

\begin{remark}
When $f_\eta$ has constant variance, the distribution of the maximum
value $f_{\eta^*}$ can be approximated extremely well by the expected   
Euler characteristic \citep{RFG} of the excursion set 
$f^{-1}_\eta(t,\infty) \cap \K$,
\begin{multline*}
 \sum_{j=0}^p \int_{\partial_j\K } \Ee \bigg(\det(-\nabla^2
 f_{|T_{\eta}\K}) \indic_{\{f_{\eta} > t, 
   \nabla f_{|(T_\eta\K)^\perp} \in N_{\eta} \K\}} \,\bigg|\, \nabla f_{|T_{\eta}\K} =0
 \bigg) \; \cdot \\
\psi_{\nabla f_{|T_{\eta}\K}}(0) \, \hauss_j(d\eta). 
\end{multline*}
This approximation is exact when $\K$ is convex
\citep{takemura:kuriki}, since the Euler characteristic of the
excursion set is equal to the indicator that it is not empty. 
\end{remark}

%
%

\subsection{Decomposition of the Hessian}

In looking at the formula \eqref{eq:dist}, we note that the quantities
\smash{$\tf^\eta_\eta, \V^-_\eta, \V^+_\eta, \V^0_\eta$} inside the
indicator are all independent, by construction, of $\nabla
f_{|T_{\eta}K}$. It will be useful to decompose the Hessian term
similarly. We write  
\begin{equation*}
-\nabla^2 f_{|T_{\eta}\K} = -H_{\eta} + G_{\eta} \cdot
\tf^{\eta}_{\eta} + R_{\eta},
\end{equation*}
where
\begin{align}
\label{eq:r}
R_{\eta} &= - \Ee_0\big(\nabla^2 f_{|T_{\eta}\K} \,\big|\, \nabla
  f_{|T_{\eta}\K}\big), \\ 
\label{eq:g}
G_{\eta} \cdot \tf^{\eta}_{\eta} &
= -\Ee_0 \big(\nabla^2 f_{|T_{\eta}\K} \,\big|\, \tf^{\eta}_{\eta} \big), \\   
\label{eq:h}
H_{\eta} &= -\big( \nabla^2 f_{|T_{\eta}\K} -  R_{\eta}  \big)-
G_{\eta} \cdot \tf^{\eta}_{\eta}. 
\end{align}
At a critical point of $f_{|\partial_j\K}$, notice that
$R_{\eta}=0$ (being a linear function of the gradient 
$\nabla f_{|T_\eta \K}$, which is zero at such a critical point). 
Furthermore, the pair of matrices \smash{$(G_{\eta}  
  \tf^{\eta}_{\eta}, H_{\eta})$} is independent of $\nabla
f_{|T_{\eta}\K}$.  Hence, we can rewrite our key formulae for the 
distribution of the maximizer and its value. 

\begin{lemma}
For each fixed $\eta\in\K$, we have 
\begin{equation*}
\tf^{\eta}_{\eta} \sim N( \mu_{\eta}, \sigma^2_{\eta}),
\end{equation*}
independently of $(\V^-_{\eta}, \V^+_{\eta}, \V^0_{\eta}, H_{\eta})$, 
with
\begin{align}
\label{eq:mu}
\mu_{\eta} &= \eta^TX(I-P_{\eta,X,\Sigma}) X\beta_0, \\
\label{eq:sigma}
\sigma^2_{\eta} &=
\eta^TX^T(I-P_{\eta,X,\Sigma})\Sigma(I-P_{\eta,X,\Sigma}^T)X\eta,
\end{align}
and (recall) \smash{$P_{\eta,X,\Sigma} = \Sigma X V_\eta(V_\eta^T X^T \Sigma
X V_\eta)^\dagger V_\eta^T X^T$}, for an orthonormal basis $V_\eta$ 
of $T_\eta\K$.

Moreover, the formula \eqref{eq:dist} can be equivalenty expressed as: 
\begin{align}
\nonumber
\Pp(\eta^*\in A, \, f_{\eta^*} \in O) 
& = \sum_{j=0}^p \int_{\partial_j\K \cap A} 
\Ee \bigg( \det(-H_{\eta} + G_{\eta} \tf^{\eta}_{\eta}) \; \cdot \\
\label{eq:dist1}
& 
\hspace{35pt}
\indic_{\{\V^-_{\eta} \leq \tf^{\eta}_{\eta} \leq \V^+_{\eta},
  \V^0_{\eta} \leq 0,\tf^{\eta}_{\eta}  \in O\}} \bigg)
\psi_{\nabla f_{|T_{\eta}\K}}(0) \, \hauss_j(d\eta) \\ 
\nonumber
 & = \sum_{j=0}^p \int_{\partial_j\K \cap A}  \Ee
 \bigg(\Mm_{\Lambda_{\eta},\V_{\eta}, \mu_{\eta},
   \sigma^2_{\eta}}(\indic_O) \indic_{\{\V^-_{\eta} \leq \V^+_{\eta},
   \V^0_{\eta} \leq 0 \}} \bigg) \; \cdot 
\\ 
\label{eq:dist2} 
& \hspace{110pt}
\psi_{\nabla f_{|T_{\eta}\K}}(0) \det(G_{\eta}) \, \hauss_j(d\eta),
\end{align}
where
\begin{equation}
\label{eq:m}
\Mm_{\Lambda,\V, \mu, \sigma^2}(h) = 
\int_{\V^-}^{\V^+} h(z) \det(\Lambda + z I)
\frac{e^{-(z-\mu)^2/2\sigma^2}}{\sqrt{2\pi\sigma^2}} \, dz,
\end{equation}
and $\Lambda_{\eta} = G_{\eta}^{-1}H_{\eta}$.
\end{lemma}

\smallskip
\begin{remark}
Until \eqref{eq:dist2}, we had not used the independence of
\smash{$\tf^\eta_\eta$} and 
\smash{$\V^-_\eta,\V^+_\eta,\V^0_\eta$} (Lemma \ref{lem:indep}).
In \eqref{eq:dist2} we do so, by first integrating over
\smash{$\tf^\eta_\eta$} 
(in the definition of $\Mm$), and then over
\smash{$\V^-_\eta,\V^+_\eta,\V^0_\eta$}.    
\end{remark}

\subsection{The conditional distribution}

The Kac-Rice formula can be generalized further. For a possibly
random function $h$, with $h_{|\partial_j\K}$ continuous for
each $j=0,\ldots p$, we see as a natural extension from
\eqref{eq:dist1},
\begin{multline}
\label{eq:disth}
\Ee \big( h(\eta^*) \big) 
= \sum_{j=0}^p \int_{\partial_j\K }  \Ee \bigg( h(\eta)
\det(-H_{\eta} + G_{\eta} \tf^{\eta}_{\eta}) \; \cdot \\
\indic_{\{\V^-_{\eta} \leq \tf^{\eta}_{\eta} \leq \V^+_{\eta},
 \V^0_{\eta} \leq 0\}} \bigg) 
\psi_{\nabla f_{|T_{\eta}\K}}(0) \, \hauss_j(d\eta). 
\end{multline}
This allows us to form a conditional distribution function of
sorts. As defined in \eqref{eq:m}, $\Mm_{\Lambda, \V, \mu, \sigma^2}$
is not a probability measure, but it can be normalized to yield
one: 
\begin{equation}
\label{eq:q}
\Qq_{\Lambda, \V, \mu, \sigma^2}(g) = \frac{\Mm_{\Lambda, \V, 
\mu, \sigma^2}(g)}{\Mm_{\Lambda, \V, \mu, \sigma^2}(1)}.  
\end{equation}
Working form \eqref{eq:disth}, with $h(\eta)=g(f_\eta)$, 
\begin{align}
\nonumber
\Ee \big( g(f_{\eta^*}) \big)
& = \sum_{j=0}^p \int_{\partial_j\K }  \Ee \bigg(
\Mm_{\Lambda_{\eta}, \V_{\eta}, \mu_{\eta}, \sigma^2_{\eta}} (g)
\indic_{\{\V^-_{\eta}   \leq \V^+_{\eta}, \V^0_{\eta} \leq 0\}}
\bigg) \; \cdot \\
\nonumber
& \hspace{150pt}
\psi_{\nabla f_{|T_{\eta}\K}}(0) \det(G_\eta) \, \hauss_j(d\eta) \\
\nonumber
& = \sum_{j=0}^p \int_{\partial_j\K }  \Ee \bigg(
\Qq_{\Lambda_{\eta}, \V_{\eta}, \mu_{\eta}, \sigma^2_{\eta}}(g) 
\Mm_{\Lambda_{\eta}, \V_{\eta}, \mu_{\eta}, \sigma^2_{\eta}}(1)
\indic_{\{\V^-_{\eta}   \leq \V^+_{\eta}, \V^0_{\eta} \leq 0\}}
\bigg) \; \cdot \\ 
\label{eq:distq}
& \hspace{150pt} \psi_{\nabla f_{|T_{\eta}\K}}(0) 
\det(G_\eta) \, \hauss_j(d\eta). 
\end{align}
This brings us to our next result. 

\begin{lemma}
\label{lem:cond} 
Formally, the measure $\Qq_{\Lambda,\V,\mu_{\eta},\sigma^2_{\eta}}$   
is a conditional distribution function of $f_{\eta^*}$, in the sense
that on each stratum $\partial_j\K$, $j=0,\ldots p$, we have
\begin{equation}
\label{eq:cond}
\Ee\big(g(f_{\eta^*}) \,\big|\, \eta^*=\eta, \,
\Lambda_{\eta^*}=\Lambda, \,
\V_{\eta^*}=\V \big) = \Qq_{\Lambda, \V, \mu_{\eta},
  \sigma^2_{\eta}}(g), 
\end{equation}
for all $\eta \in \partial_j\K$.
\end{lemma}

\begin{proof}
By expanding the right-hand side in \eqref{eq:distq}, we see
\begin{multline*}
\Ee\big(g(f_{\eta^*})\big) =
\sum_{j=0}^p \int_{\partial_j\K } \int_{\real^j \times \real^3} 
\Qq_{\Lambda, \V, \mu_{\eta}, \sigma^2_{\eta}}(g) 
\Mm_{\Lambda, \V, \mu_{\eta}, \sigma^2_{\eta}}(1) 
\indic_{\{\V^-   \leq \V^+, \V^0 \leq 0\}} \; \cdot \\
\psi_{\nabla f_{|T_{\eta}\K}}(0) \det(G_\eta) 
\, F_{\Lambda_\eta, \V_\eta}(d\Lambda,d\V)
\, \hauss_j(d\eta)
\end{multline*}
where $F_{\Lambda_\eta,\V_\eta}$ denotes the joint distribution of
$(\Lambda_\eta,\V_\eta) =
(\Lambda_\eta,\V^-_\eta,\V^+_\eta,\V^0_\eta)$ at a fixed value of
$\eta$. Consider the quantity 
\begin{equation}
\label{eq:dens}
\Mm_{\Lambda, \V, \mu_{\eta}, \sigma^2_{\eta}}(1) 
\indic_{\{\V^-   \leq \V^+, \V^0 \leq 0\}} 
\psi_{\nabla f_{|T_{\eta}\K}}(0) \det(G_\eta) 
\, F_{\Lambda_\eta, \V_\eta}(d\Lambda,d\V)
\, \hauss_j(d\eta).
\end{equation}
We claim that this is the joint density (modulo differential
terms) of $\eta^*,\Lambda_{\eta^*},\V_{\eta^*}$, when $\eta^*$ is
restricted to  
the smooth piece $\partial_j\K$.  This would complete the proof. 
Hence to verify the claim, we the apply Kac-Rice formula with open
sets $A \subseteq \partial_j\K$, $B \subseteq \real^{j\times j}$, and
$C \subseteq \real^3$, giving
\begin{align*}
\Pp ( \eta^* \in A, \, \Lambda_{\eta^*} \in B, \,
\V_{\eta^*} \in C) &= 
\int_A \Ee 
\bigg(\Mm_{\Lambda_{\eta}, \V_{\eta}, \mu_{\eta}, \sigma^2_{\eta}}(1)
\indic_{\{\V^-_{\eta}   \leq \V^+_{\eta}, \V^0_{\eta} \leq 0\}} \;
\cdot \\
& \hspace{40pt}
\indic_{\{\Lambda \in B, \V \in C\}} \bigg) 
\psi_{\nabla f_{|T_{\eta}\K}}(0) \det(G_\eta) \, \hauss_j(d\eta) \\
&= \int_A \int_{B \times C}
\Mm_{\Lambda, \V, \mu_{\eta}, \sigma^2_{\eta}}(1) 
\indic_{\{\V^-   \leq \V^+, \V^0 \leq 0\}} \; \cdot \\
&\hspace{30pt}
\psi_{\nabla f_{|T_{\eta}\K}}(0) \det(G_\eta) 
\, F_{\Lambda_\eta, \V_\eta}(d\Lambda,d\V)
\, \hauss_j(d\eta). 
\end{align*}
Taking $A,B,C$ to be open balls around some fixed points
$\eta,\Lambda,\V$, respectively, and sending their radii to zero, we
see that the joint density of $\eta^*,\Lambda_{\eta^*},\V_{\eta^*}$,
with $\eta^*$ restricted to $\partial_j\K$, is exactly as in
\eqref{eq:dens}, as desired.  
\end{proof}

\begin{remark}
\label{rem:uncond}
The analogous result also holds unconditionally, i.e., it is clear
that 
\begin{equation*}
\Ee\big(g(f_{\eta^*})\big) =
\Ee\Big(\Qq_{\Lambda_{\eta^*},\V_{\eta^*},\mu_{\eta^*},\sigma^2_{\eta^*}}(g)\Big),
\end{equation*}
by taking an expectation on both sides of \eqref{eq:cond}.
\end{remark}






\subsection{The Kac-Rice pivotal quantity}
\label{sec:pval}

Suppose that we are interested in testing the null hypothesis $H_0 : y
\sim N(0,\Sigma)$.  We might look at the observed value of the first
knot $\lambda_1=f_{\eta^*}$, and see if it was larger than we would
expect under $H_0$. From the results of the last section,
\begin{equation*}
\Pp(f_{\eta^*} > t)  =
\Ee\Big(\Qq_{\Lambda_{\eta^*},\V_{\eta^*},\mu_{\eta^*},\sigma^2_{\eta^*}}
(\indic_{(t,\infty)})\Big),
\end{equation*}
and so the most natural strategy seems to be to plug our
observed value of the first knot into the above formula.  This,
however, requires computing the above expectation, i.e., the integral
in \eqref{eq:distq}.

In this section, we present an alternative approach that is
effectively a conditional test, conditioning on the observed value
of $\eta^*$, as well as $\Lambda_{\eta^*}$ and $\V_{\eta^*}$.  
To motivate our test, it helps to take a step back and think about the
measure $\Qq_{\Lambda,\V,\mu,\sigma^2}$ defined in
\eqref{eq:q}.  For fixed values of $\Lambda,\V,\mu,\sigma^2$, we can
reexpress this (nonrandom) measure as
\begin{equation*}
\Qq_{\Lambda,\V,\mu,\sigma^2}(g) = \int_{-\infty}^\infty g(t) \cdot
q_{\Lambda,\V,\mu,\sigma}(t) \,dt,
\end{equation*}
where $q_{\Lambda,\V,\mu,\sigma}$ is a density function (supported on
$[\V^-,\V^+]$). 
 In other words, $\Qq_{\Lambda,\V,\mu,\sigma^2}(g)$
computes the expectation of $g$ with respect to a density
$q_{\Lambda,\V,\mu,\sigma^2}$, so we can write
$\Qq_{\Lambda,\V,\mu,\sigma^2}(g) = \Ee(g(W))$ where $W$ is a random
variable whose density is $q_{\Lambda,\V,\mu,\sigma^2}$.  Now consider
the survival function 
\begin{equation*}
\Ss_{\Lambda, \V, \mu, \sigma^2}(t) = 
\Qq_{\Lambda, \V, \mu, \sigma^2}(\indic_{(t,\infty)}) = \Pp(W > t). 
\end{equation*}
A classic argument shows that $\Ss_{\Lambda,\V,\mu,\sigma^2}(W) \sim
\mathrm{Unif}(0,1)$.  Why is this useful?  Well, according to lemma
\ref{lem:cond} (or, Remark \ref{rem:uncond} following the lemma), the
first knot $\lambda_1=f_{\eta^*}$ almost takes the role of $W$ above,
except that there is a further level of randomness in $\eta^*$, and 
$\Lambda_{\eta^*},\V_{\eta^*}$.  That is, instead of the expectation
of $g(f_{\eta^*})$ being given by
\smash{$\Qq_{\Lambda_{\eta^*},\V_{\eta^*},\mu_{\eta^*},\sigma^2_{\eta^*}}(g)$},
it is given by 
\smash{$\Ee(\Qq_{\Lambda_{\eta^*},\V_{\eta^*},\mu_{\eta^*},\sigma^2_{\eta^*}}(g))$}. 
The key intuition is that the random variable
\begin{equation}
\label{eq:surv}
\Ss_{\Lambda_{\eta^*},\V_{\eta^*},\mu_{\eta^*},\sigma^2_{\eta^*}}
(f_{\eta^*}) =
\Qq_{\Lambda_{\eta^*},\V_{\eta^*},\mu_{\eta^*},\sigma^2_{\eta^*}}
(\indic_{(f_{\eta^*},\infty)})
\end{equation}
should still be uniformly distributed, since this is true
conditional on $\eta^*,\Lambda_{\eta^*},\V_{\eta^*}$, and
unconditionally, the extra level of randomness in
$\eta^*,\Lambda_{\eta^*},\V_{\eta^*}$ just gets ``averaged out'' and
does not change the distribution. Our next lemma formalizes this
intuition, and therefore provides a test for $H_0$ based on the
(random) survival function in \eqref{eq:surv}. 

\begin{lemma}
\label{lem:mainresult}[Kac-Rice pivot]
The survival function of $\Qq_{\Lambda,\V,\mu,\sigma^2}$, with
$\Lambda=\Lambda_{\eta^*}$, $\V=\V_{\eta^*}$, $\mu=\mu_{\eta^*}$,
$\sigma^2=\sigma^2_{\eta^*}$, and evaluated at $t=f_{\eta^*}$,
satisfies 
\begin{equation}
\label{eq:sunif2}
\Ss_{\Lambda_{\eta^*}, \V_{\eta^*}, \mu_{\eta^*},
  \sigma^2_{\eta^*}}(f_{\eta^*}) \sim \mathrm{Unif}(0,1).
\end{equation}
\end{lemma}

\begin{proof}
Fix some $h:\real \rightarrow \real$.  A standard argument shows that 
(fixing $\Lambda,\V,\mu,\sigma^2$),
\begin{equation*}
\Qq_{\Lambda, \V, \mu, \sigma^2}
\big(h \circ \Ss_{\Lambda, \V, \mu, \sigma^2} \big) = 
\int_0^1 h(t) \, dt.
\end{equation*}
Now we compute, applying \eqref{eq:distq} with $g$ being a
composition of functions,
\begin{align*}
\Ee \bigg(h \Big(&S_{\Lambda_{\eta^*}, \V_{\eta^*}, \mu_{\eta^*},
  \sigma^2_{\eta^*}}(f_{\eta^*}) \Big) \bigg) \\
& = \sum_{j=0}^p \int_{\partial_j\K}  \Ee \bigg(
\Qq_{\Lambda_{\eta}, \V_{\eta}, \mu_{\eta}, \sigma^2_{\eta}} \Big(h
\circ S_{\Lambda_{\eta}, \V_{\eta}, \mu_{\eta}, \sigma^2_{\eta}} \Big) 
\Mm_{\Lambda_{\eta}, \V_{\eta}, \mu_{\eta}, 
  \sigma^2_{\eta}}(1)  \; \cdot \\
& \hspace{105pt} \indic_{\{\V^-_{\eta} \leq \V^+_{\eta},
  \V^0_{\eta} \leq 0\}} \bigg)  \psi_{\nabla f_{|T_{\eta}\K}}(0) 
\det(G_\eta) \, \hauss_j(d\eta) \\ 
& = \sum_{j=0}^p \int_{\partial_j\K}  \Ee \bigg(
\Big[\int_0^1 h(t) \, dt \Big]
\Mm_{\Lambda_{\eta}, \V_{\eta}, \mu_{\eta}, 
  \sigma^2_{\eta}}(1)  \indic_{\{\V^-_{\eta} \leq \V^+_{\eta},
  \V^0_{\eta} \leq 0\}} \bigg) \; \cdot \\
& \hspace{175pt} \psi_{\nabla f_{|T_{\eta}\K}}(0) 
\det(G_\eta) \, \hauss_j(d\eta) \\ 
& = 
\Big[\int_0^1 h(t) \, dt \Big]
\sum_{j=0}^p \int_{\partial_j\K}  \Ee \Big(
\Mm_{\Lambda_{\eta}, \V_{\eta}, \mu_{\eta}, 
  \sigma^2_{\eta}}(1)  \indic_{\{\V^-_{\eta} \leq \V^+_{\eta},
  \V^0_{\eta} \leq 0\}} \Big) \; \cdot \\
& \hspace{175pt} \psi_{\nabla f_{|T_{\eta}\K}}(0) 
\det(G_\eta) \, \hauss_j(d\eta) \\ 
&=  \int_0^1 h(t) \, dt.
\end{align*}
\end{proof}

\begin{remark}
In particular, Lemma \ref{lem:mainresult} shows that under $H_0$, 
\begin{equation*}
\Ss_{\Lambda_{\eta^*}, \V_{\eta^*}, 0,
\sigma^2_{\eta^*}}(f_{\eta^*}) \sim \mathrm{Unif}(0,1).
\end{equation*} 
This proves our main result, Theorem \ref{thm:mainresult}, noting
that the statistic in \eqref{eq:sunif1} is just 
\smash{$\Ss_{\Lambda_{\eta^*}, \V_{\eta^*}, 0,
\sigma^2_{\eta^*}}(f_{\eta^*})$} written out a little more explicitly.
\end{remark}

\begin{remark}
We have used the survival function of $f_{\eta}$ conditional on
$\eta$ being the global maximizer as well the additional local information
$(\Lambda_{\eta}, \V_{\eta})$. Conditioning on this
extra information makes the test very simple to compute, at least
in the LASSO, group LASSO and nuclear norm cases. If we were
to marginalize over these quantities, we would have a more powerful test. In general,
it seems difficult to analytically marginalize over these quantities, but perhaps Monte
Carlo schemes would be feasible. Further, Lemma \ref{lem:mainresult} holds for
any $\mu$, implying that this marginalization over $(\Lambda_{\eta}, \V_{\eta})$ would need
access to the unknown $\mu$. Under the global null,  $H_0:\mu=0$, this is not
too much of an issue, though it already causes a problem for the construction
of selection intervals described in Section \ref{sec:selectiveinf}.
\end{remark}

\section{Practicalities and examples}
\label{sec:examples}

Given an instance of the regularized regression problem in
\eqref{eq:genprob}, we seek to compute the test statistic
\begin{equation}
\label{eq:teststat}
\Ss_{\Lambda_{\eta^*}, \V_{\eta^*}, 0,
  \sigma^2_{\eta^*}}(\lambda_1) = 
\frac{\displaystyle \int_{\lambda_1}^{\V^+_{\eta^*}}
\det(\Lambda_{\eta^*}+z I) \phi_{\sigma^2_{\eta^*}} (z) \, dz}
{\displaystyle \int_{\V^-_{\eta^*}}^{\V^+_{\eta^*}}
\det(\Lambda_{\eta^*}+z I) \phi_{\sigma^2_{\eta^*}} (z) \, dz},
\end{equation}
and compare this against $\mathrm{Unif}(0,1)$.
Recalling that $\Lambda_\eta = H_\eta^{-1} G_\eta$, this leaves 
us with essentially 6 quantities to be computed---\smash{$\lambda_1, 
\V^{+}_{\eta^*}, \V^{-}_{\eta^*}, G_{\eta^*},
H_{\eta^*},\sigma^2_{\eta^*}$}---and the above integral to be
calculated.  

If we know the dual seminorm $\dualpen$ of the penalty
$\pen$ in closed form, then the first knot $\lambda_1$ can be found
explicitly, as in $\lambda_1 = \dualpen(X^T(I-P_{X\dualball^\perp})y)$;
otherwise, it can be found numerically by solving the (convex)
optimization problem   
\begin{equation*}
\lambda_1 = \max_{\eta \in \real^p} \eta^T X^T (I-P_{X\dualball^\perp}) y \;\;
\st \;\; \pen(\eta) \leq 1.
\end{equation*}
The remaining quantities,
\smash{$\V^{+}_{\eta^*}, \V^{-}_{\eta^*}, G_{\eta^*},
H_{\eta^*},\sigma^2_{\eta^*}$}, all depend on $\eta^*$ and on the
tangent space $T_{\eta^*}\K$.  Again, depending on $\dualpen$, the
maximizer $\eta^*$ can either be found in closed form, or numerically
by solving the above optimization problem.  Once we know the
projection operator onto the tangent space $T_{\eta^*}\K$, there is an
explicit expression for $\sigma^2_{\eta^*}$, recall \eqref{eq:sigma}; 
furthermore, $\V^-_{\eta^*},\V^+_{\eta^*}$ are given by two 
more tractable (convex) optimization problems (which in some cases
admit closed form solutions), see Section \ref{sec:linfrac}.   

The quantities $G_{\eta^*},H_{\eta^*}$ are different, however; even
once we know $\eta^*$ and the tangent space $T_{\eta^*}\K$, finding  
$G_{\eta^*},H_{\eta^*}$ involves computing the Hessian $\nabla^2
f_{|T_{\eta^*}\K}$, which requires a geometric understanding of the
curvature of $f$ around $T_{\eta^*}\K$.  That is,
$G_{\eta^*},H_{\eta^*}$ cannot be calculated numerically (say,
via an optimization procedure, as with
$\lambda_1,\V_{\eta^*}^-,\V^+_{\eta^*}$), and demand a more
problem-specific, mathematical focus. For this reason, computation of 
$G_{\eta^*},H_{\eta^*}$ can end up being an involved process
(depending on the problem).  In the examples that follow, we do not
give derivation details for the Hessian $\nabla^2 f_{|T_{\eta^*}\K}$,
but refer the reader to \citet{RFG} for the appropriate background 
material. 

\JLcomment{This seemed like a natural spot to include the algorithm}

\begin{algorithm}
  \caption{Computing the \textit{Kac-Rice} pivot}
  \label{algo:fs}
  \begin{algorithmic}[1]
    \STATE Solve for $\lambda_1$ and $\eta^*$ \COMMENT{using consistent notation, see \eqref{eq:f} and Section \ref{sec:notation}}
    \STATE Form an orthonormal basis $V_{\eta^*}$ of the tangent space $T_{\eta^*}\K$
    \STATE Compute the projection $P_{\eta^*,X,\Sigma}$ in \eqref{eq:projection}
    \STATE Evaluate the conditional variance $\sigma^2_{\eta^*}$ and $C_{X,\Sigma}(\eta^*)$ from \eqref{eq:sigma} and \eqref{eq:c}
    \IF{$f_{|T_{\eta^*}\K}$ has zero Hessian}
        \STATE Let $\Lambda_{\eta^*} = 0$
    \ELSIF{$\nabla^2 f_{|T_{\eta^*}\K} \neq 0$}
        \STATE Let $\Lambda_{\eta^*}=G^{-1}_{\eta^*} H_{\eta^*}$ from \eqref{eq:g} and \eqref{eq:h}
    \ENDIF
    \STATE Solve the optimization problems \eqref{eq:vminus} and \eqref{eq:vplus}, yielding $\V^-_{\eta^*}$, $\V^+_{\eta^*}$
    \STATE Evaluate the integrals in \eqref{eq:teststat} to obtain $\Ss = \Ss_{\Lambda_{\eta^*}, \V_{\eta^*}, 0, \sigma^2_{\eta^*}}(\lambda_1)$
    \RETURN $\Ss$
  \end{algorithmic}
\end{algorithm}

We now revisit the lasso example, and then consider the 
group lasso and nuclear norm penalties, the latter yielding
applications to principal components and matrix completion.   We
remark that in the lasso and group lasso cases, the matrix
$\Lambda_{\eta^*}=G^{-1}_{\eta^*} H_{\eta^*}$ is zero, simplifying the
computations. In contrast, it is nonzero for the nuclear norm case.    

Also, it is important to point out that in all three problem cases, we
have $\spa(C)=\real^p$, so the notational shortcut that we applied
in Sections \ref{sec:character} and \ref{sec:kacrice} has no effect
(see Section \ref{sec:notation}), and we can use the formulae from
these sections as written.

\subsection{Example: the lasso (revisited)}

For the lasso problem \eqref{eq:lasso}, we have
$\pen(\beta)=\|\beta\|_1$ and 
$\dualball = \{u : \|u\|_\infty \leq 1\}$, so 
$\dualpen(\beta)=\|\beta\|_\infty$ and 
$\K = \dualball^\circ = \{v : \|v\|_1 \leq 1\}$.
Our Morse assumption on the process $f_{\eta}=\eta^T X^T y$ over $\K$ (which amounts
to an assumption on the design matrix $X$) 
implies that there is a unique index $j^*$ such that 
\begin{equation*}
\lambda_1 = |X_{j^*}^T y| = 
\|X^T y\|_\infty = \max_{\|\eta\|_1 \leq 1} \eta^T X^T y.
\end{equation*}
Then in this notation $\eta^*=\sign(X_{j^*}^T y) \cdot e_{j^*}$  
(where $e_{j^*}$ is the
$j^*$th standard basis vector), and the normal cone to $\K$ at
$\eta^*$ is  
\begin{equation*}
N_{\eta^*}\K = \{v \in \real^p : \sign(v_{j^*})=\sign(X_{j^*}^Ty), \;
|v_j| \leq |v_{j^*}| \;\, \text{for all} \;\, j\not=j^* \}.
\end{equation*}
Because this is a full-dimensional set, the tangent space to $\K$ at 
$\eta^*$ is $T_{\eta^*}\K = (N_{\eta^*}\K)^\perp = \{0\}$.  This
greatly simplifies our survival function test statistic \eqref{eq:surv} since
all matrices in consideration here are $0\times 0$ and 
therefore have determinant 1, giving
\begin{equation*}
\Ss_{\Lambda_{\eta^*}, \V_{\eta^*}, 0, \sigma^2_{\eta^*}} = 
\frac{\Phi(\V_{\eta^*}^-/\sigma_{\eta^*}) -
  \Phi(\lambda_1/\sigma_{\eta^*})}
{\Phi(\V_{\eta^*}^-/\sigma_{\eta^*}) -
  \Phi(\V_{\eta^*}^+/\sigma_{\eta^*})}
\end{equation*}
The lower and upper limits $\V_{\eta^*}^+,\V_{\eta^*}^-$ are easily
computed by solving two linear fractional programs, see Section
\ref{sec:linfrac}.  The variance $\sigma^2_{\eta^*}$ of $f_{\eta^*}$
is given by \eqref{eq:sigma}, and again simplifies because
$T_{\eta^*}\K$ is zero dimensional, becoming
\begin{equation*}
\sigma^2_{\eta^*} = (\eta^*)^T X^T \Sigma X \eta^* = X_{j^*}^T \Sigma
X_{j^*}. 
\end{equation*}
Plugging in this value gives the test 
statistic as in \eqref{eq:lassopval1} in Section \ref{sec:lasso}.  The
reader can return to this section for examples and discussion in the
lasso case.

\subsection{Example: the group lasso}
\label{sec:group}

The group lasso \citep{grouplasso} can be viewed
as an extension of the lasso for grouped (rather than individual)
variable selection.  Given a pre-defined collection $\cG$ of groups,
with $\cup_{g \in \cG} g = \{1,\ldots p\}$,
the group lasso penalty is defined as
\begin{equation*}
\pen(\beta) = \sum_{g=1}^G w_g \|\beta_g\|_2, 
\end{equation*}
where $\beta_g \in \real^{|g|}$ denotes the subset of
components of $\beta \in \real^p$ corresponding to $g$, and $w_g > 0$ 
for all $g \in \cG$. 
We note that
\begin{equation*}
\dualball = \{u \in \real^p : \|u_g\|_2 \leq w_g, \; g \in \cG\},
\end{equation*}
so the dual of the penalty is 
\begin{equation*}
\dualpen(\beta) = \max_{g \in \cG} \, w_g^{-1} \|\beta_g\|_2,
\end{equation*}
and 
\begin{equation*}
\K = \dualball^\circ = \Big\{v \in \real^p : \sum_{g\in\cG} w_g \|v_g\|_2
\leq 1\Big\}.
\end{equation*}
Under the Morse assumption on $f_\eta = \eta^T X^T y$ over $\K$ (again, this
corresponds to an assumption about the design matrix $X$), there
is a unique group $g^*$ such that 
\begin{equation*}
\lambda_1 = w_{g^*}^{-1} \|X_{g^*}^T y\|_2 =
\max_{g \in \cG} \, w_g^{-1} \|X_g^T y\|_2 = 
\max_{\sum_{g\in \cG} w_g \|\eta_g\|_2 \leq 1} \, \eta^T X^T y,
\end{equation*}
where we write $X_g \in \real^{n\times |g|}$ to denote the matrix
whose columns are a subset of those of $X$, corresponding to $g$.   
Then the maximizer $\eta^*$ is given by
\begin{equation*}
\eta^*_g = \begin{cases}
\displaystyle
\frac{X_g^Ty}{w_g\|X_g^Ty\|_2} &
\text{if}\;\, g=g^* \\
0 & \text{otherwise}
\end{cases},
\qquad 
\text{for all}\;\, g \in \cG,
\end{equation*}
and the normal cone $N_{\eta^*}\K$ is seen to be
\begin{multline*}
N_{\eta^*}\K = \Big\{v\in \real^p : 
v_{g^*} = c \,X_{g^*}^T y, \;
\|v_g\|_2/w_g \leq c \|X_{g^*}^T y\|_2 / w_{g^*} \\
\text{for all}\;\, g\not=g^*, \; c\geq 0 \Big\}.
\end{multline*}
Hence the tangent space $T_{\eta^*}\K = (N_{\eta^*}\K)^\perp$ is 
\begin{equation*}
T_{\eta^*}\K = 
\Big\{ u \in \real^p : u_{g^*}^T X_{g^*}^T y = 0, \;
u_g = 0 \;\,\text{for all}\;\, g\not=g^*\Big\},
\end{equation*}
which has dimension $r^*-1$, with $r^*=\mathrm{rank}(X_{g^*})$. 
An orthonormal basis $V_{\eta^*}$ for this tangent space is 
given by padding an orthonormal basis for $(\spa(X_{g^*}^T y))^\perp$ 
with zeros appropriately.  From this we can compute the projection
operator 
\begin{equation*}
P_{\eta^*,X,\Sigma} = \Sigma X V_{\eta^*} (V_{\eta^*}^T X^T \Sigma X 
V_{\eta^*})^\dagger V_{\eta^*}^T X^T,
\end{equation*}
and the variance of $f_{\eta^*}$ as
\begin{equation*}
\sigma_{\eta^*}^2 = \frac{1}{w_{g^*}^2 \|X_{g^*}^T y\|_2^2}
y^T X_{g^*}X_{g^*}^T (I-P_{\eta^*,X,\Sigma}) \Sigma X_{g^*}X_{g^*}^T
y. 
\end{equation*}
The quantities $\V^-_{\eta^*},\V^+_{\eta^*}$ can be readily computed by
solving two convex programs, see Section \ref{sec:linfrac}.
Finally, we have $H_{\eta^*}=0$ in the group lasso
case, as the special form of curvature matrix of a sphere
implies that
 \smash{$G_{\eta^*} \tf^{\eta^*}_{\eta^*} = -\nabla^2
   f_{|T_{\eta^*}\K}$} in \eqref{eq:g}. 
This makes $\Lambda_{\eta^*}=G_{\eta^*}^{-1}H_{\eta^*}=0$, and the
test statistic \eqref{eq:teststat} for the group lasso problem becomes 
\begin{equation}
\label{eq:grouppval}
\frac{\displaystyle \int_{\lambda_1}^{\V^+_{\eta*}} z^{r^*-1}
  \phi_{\sigma^2_{\eta^*}}(z) \, dz}
{\displaystyle \int_{\V^-_{\eta^*}}^{\V^+_{\eta*}} z^{r^*-1}
  \phi_{\sigma^2_{\eta^*}}(z) \, dz} = 
\frac{  \Pp(\chi_{r^*} \leq \V^+_{\eta^*}/\sigma_{\eta^*})  -  
\Pp(\chi_{r^*} \leq \lambda_1 / \sigma_{\eta^*})}
{\Pp (\chi_{r^*} \leq \V^+_{\eta^*} / \sigma_{\eta^*})
 - \Pp (\chi_{r^*} \leq \V^-_{\eta^*}/\sigma_{\eta^*})}.
\end{equation}
In the above, $\chi_{r^*}$ denotes a chi distributed random
variable with $r^*$ degrees of freedom, and the equality follows from
the fact that the missing multiplicative factor in the $\chi_{r^*}$
density [namely, $2^{1-r^*/2}/\Gamma(r^*/2)$] is common to the 
numerator and denominator, and hence cancels.

\begin{figure}
\begin{center}
\subfigure[Kac-Rice test]{
\label{fig:grouplasso:pval:exact}
\includegraphics[width=0.5\textwidth]{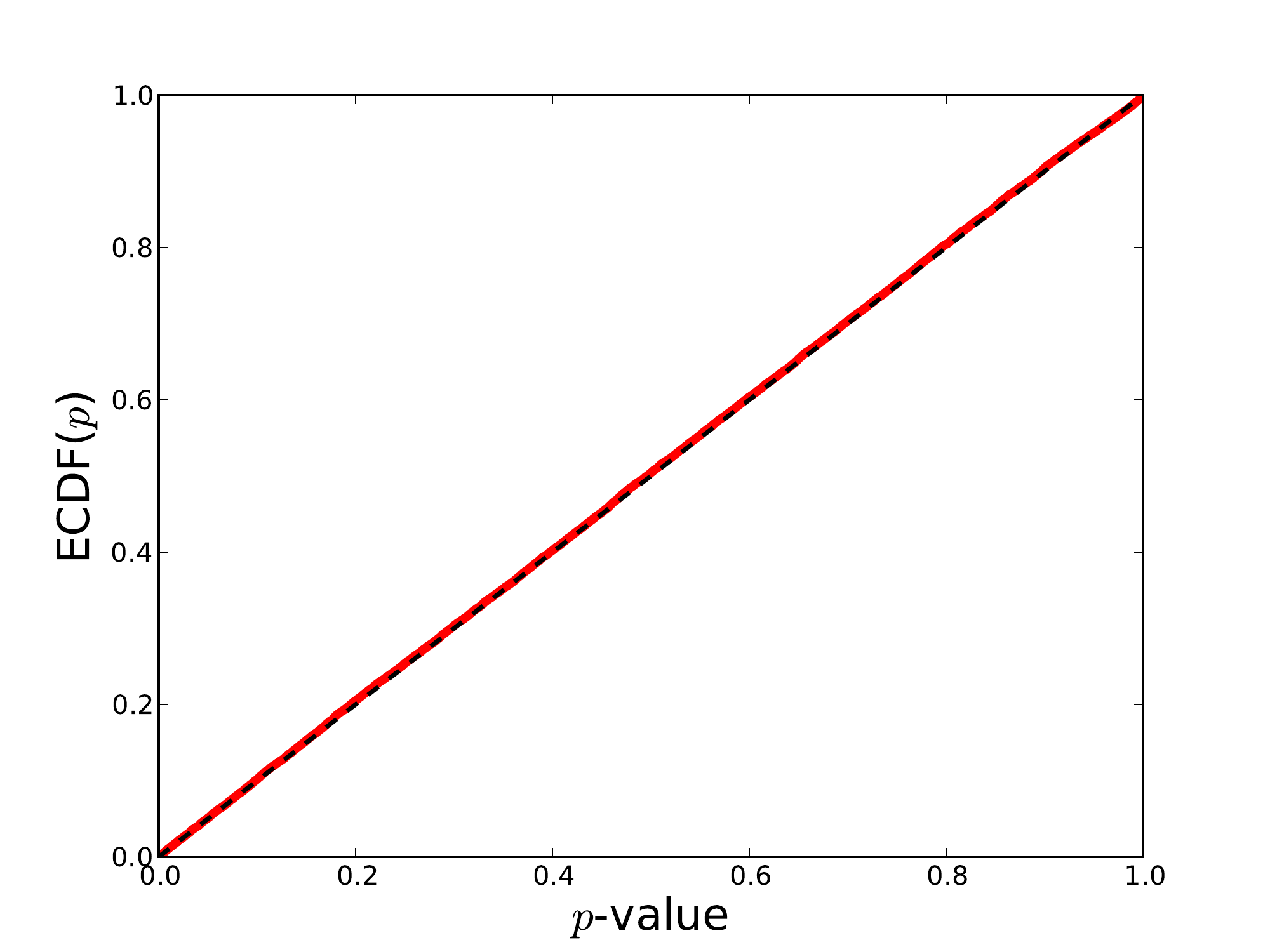}}
\hspace{-15pt}
\subfigure[Covariance test]{
\label{fig:grouplasso:pval:exp}
\includegraphics[width=0.5\textwidth]{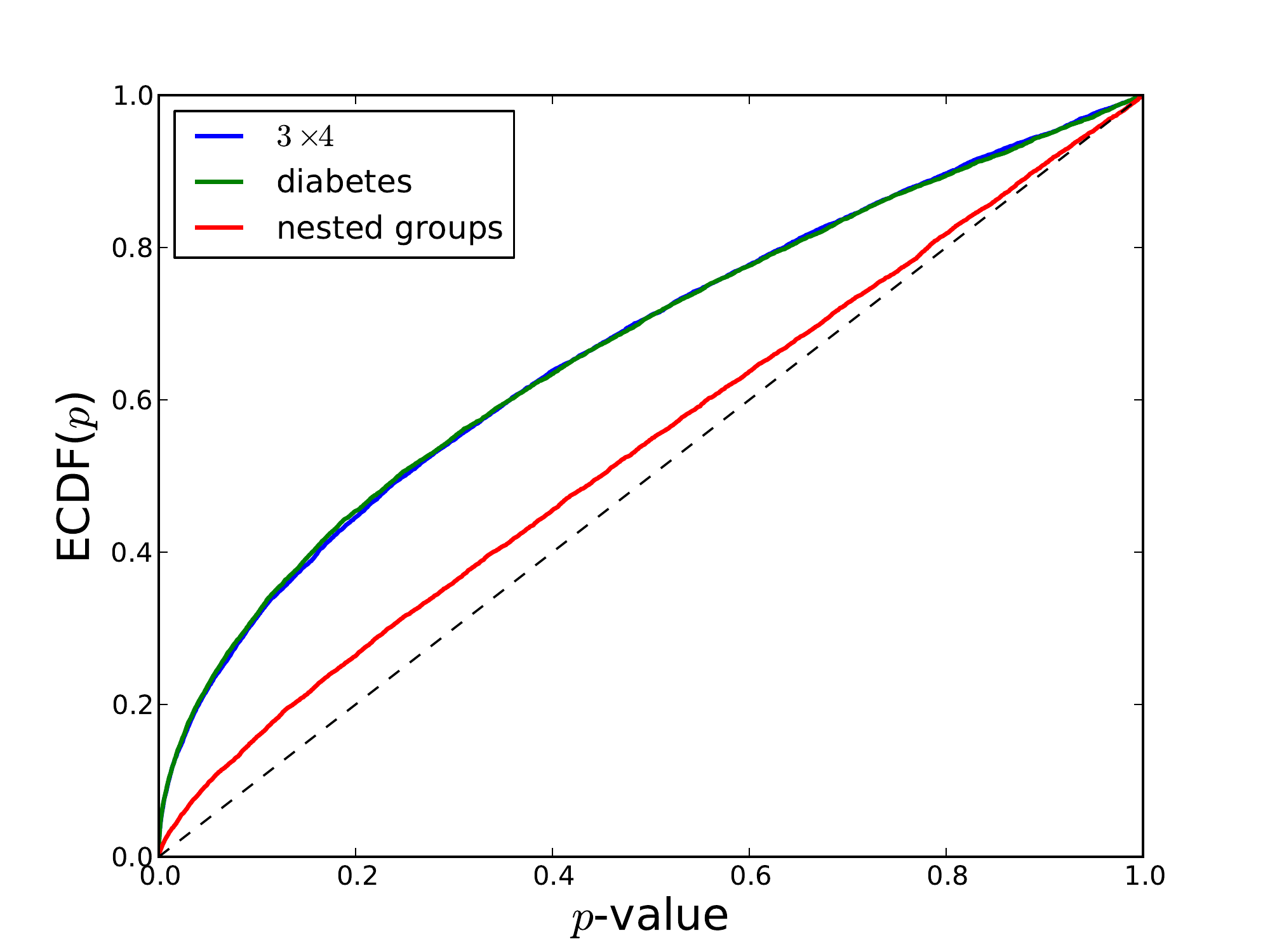}}
\caption{\small \it The left panel shows the empirical
    distribution function of a sample of 20,000 p-values
    \eqref{eq:grouppval} computed from various group lasso setups,
    which agrees very closely with the uniform distribution.  The
    right panel shows the empirical distribution functions in three
    different group lasso setups, each over 10,000 samples, when
    p-values are instead computed using an $\mathrm{Exp}(1)$
    approximation for the covariance test.  This approximation 
    ends up being anti-conservative whereas the Kac-Rice test is exact.}
\end{center}
\end{figure}

Figure \ref{fig:grouplasso:pval:exact} shows the empirical
distribution function of a sample of 20,000 p-values from problem
instances sampled randomly from a variety of different group lasso
setups (all under the global null model $\beta_0=0$):
\begin{itemize}
\item small case: $X$ is $3 \times 4$, a fixed matrix slightly
  perturbed by Gaussian noise; there are 2 groups of size 2, one with 
  weight $\sqrt{2}$, the other with weight $0.1$;
\item fat case: $X$ is $100 \times 10,000$ with features drawn from
  the compound symmetric Gaussian distribution having correlation
  0.5; here are 1000 groups each of size 10, each having weight
  $\sqrt{10}$;
\item tall case: $X$ is $10,000 \times 100$ with features drawn from
  the compound symmetric Gaussian distribution having correlation 0.5; 
  there are 1000 groups each of size 10, each having weight $\sqrt{10}$;
\item square case: $X$ is $100 \times 100$ with features drawn from
  the compound symmetric Gaussian distribution having correlation 0.5;
  there are 10 groups each of size 10, each having weight $\sqrt{10}$;
\item diabetes case 1: $X$ is $442 \times 10$, the diabetes data set
  from \citet{lars}; there are 4 (arbitrarily created) groups: one of
  size 4, one of size 2, one of size 3 and one of size 1, with varying
  weights;
\item diabetes case 2: $X$ is $442 \times 10$, the diabetes data set
  from \citet{lars}; there are now 10 groups of size 1 with i.i.d. random
  weights drawn from $1 + 0.2 \cdot \mathrm{Unif}(0,1)$ (generated
  once for the entire simulation);
\item nested case 1: $X$ is $100 \times 10$, with two nested groups
  (the column space for one group of size 2 is contained in that of
  the other group of size 8) with the weights favoring inclusion of
  the larger group first; 
\item nested case 2: $X$ is $100 \times 10$, with two nested groups
  (the column space for one group of size 2 is contained in that of
  the other group of size 8) with the weights favoring inclusion of
  the smaller group first;
\item nested case 3: $X$ is $100 \times 12$, with two sets of two
  nested groups (in each set, the column space for one group of size 2
  is contained in that of the other group of size 4) with the weights
  chosen according to group size; 
\item nested case 4: $X$ is $100 \times 120$, with twenty sets of two
  nested groups (in each set, the column space for one group of size 2
  is contained in that of the other group of size 4) with the weights
  chosen according to group size.
\end{itemize}
As we can see from the plot, the p-values are extremely close to
uniform.  

In comparison, arguments similar to those given in \citet{covtest} for
the lasso case would suggest that for the group lasso, under the null
hypothesis, 
\begin{equation*}
\frac{ \lambda_1 (\lambda_1 - \V^-_{\eta^*})}
{\sigma_{\eta^*}^2} \overset{d}{\to} \text{Exp}(1)
\quad \text{as} \quad n,p \rightarrow \infty,
\end{equation*}
under some conditions (one of these being that $\V^-_{\eta^*}$
diverges to $\infty$ fast enough).  
Figure \ref{fig:grouplasso:pval:exp} shows the  empirical distribution
function of 10,000 samples from three of the above scenarios,
demonstrating that, while asymptotically reasonable,
the $\text{Exp}(1)$ approximation for the covariance test in the group 
lasso case can be quite anti-conservative in finite samples.

\subsection{Nuclear norm}
\label{sec:nuclear}

In this setting, we treat the coefficients in \eqref{eq:genprob} as a 
matrix, instead of a vector, denoted by $B \in 
\real^{n \times p}$.  We consider a nuclear norm penalty on $B$,
\begin{equation*}
\pen(B) = \|B\|_* = \tr(D),
\end{equation*}
where $D$ is the diagonal matrix of singular values in the singular
value decomposition $B=UDV^T$.  Here the dual seminorm is 
\begin{equation*}
\dualpen(B) = \|B\|_\op = \max(D),
\end{equation*}
the operator norm (or spectral norm) of $B$, i.e., its maximum
singular value.  Therefore we have
\begin{align*}
\dualball &= \{ A : \|A\|_\op \leq 1\}, \\
\K = \dualball^\circ &= \{ W : \|W\|_* \leq 1\}.
\end{align*}
Examples of problems of the form \eqref{eq:genprob} with nuclear norm
penalty $\pen(B)=\|B\|_*$ can be roughly categorized according to the
choice of linear operator $X=X(B)$.  For example, 
\begin{itemize}
\item {\it principal components analysis:} if $X:\real^{n \times p}
\rightarrow \real^{n \times p}$ is the identity map, then $\lambda_1$
is the largest singular value of $y \in \real^{n \times p}$, and
moreover, $\V^-_{\eta^*}$ is the second largest singular value of $y$; 
\item {\it matrix completion:} if $X:\real^{n\times p} \rightarrow
\real^{n\times p}$ zeros out all of the entries of its argument
outside some index set $\mathcal{O} \subseteq \{1,\ldots n\} \times 
\{1,\ldots p\}$, and leaves the entries in $\mathcal{O}$ untouched,
then problem \eqref{eq:genprob} is a noisy version of the matrix
completion problem
\citep{candes:recht:exact,mazumder:hastie:softsvd}. 
\item {\it reduced rank regression:} if $X:\real^{n\times p}
  \rightarrow \real^{m\times p}$ performs matrix multiplication,
  $X(B)=XB$, and $y\in\real^{m\times p}$, then problem
  \eqref{eq:genprob} is often referred to as reduced rank 
  regression \citep{reducedrank}.
\end{itemize}
The first knot in the solution path is given by
$\lambda_1 = \|X^T(y)\|_\op$, with $X^T$ denoting the adjoint of the
linear operator $X$.  
Assuming that $X^T(y)$ has singular value decomposition $X^T(y)=UDV^T$ 
with $D=\mathrm{diag}(d_1,d_2,\ldots)$ for $d_1 \geq d_2 \geq \ldots$,
and that the process $f_\eta=\langle \eta, X^T(y)\rangle$ is Morse
over $\eta \in \K$, there is a unique $\eta^*\in\K$ achieving the
value $\lambda_1$,  
\begin{equation*}
\eta^* = U_1 V_1^T,
\end{equation*}
where $U_1,V_1$ are the first columns of $U,V$, respectively.  The
normal cone $N_{\eta^*}\K$ is
\begin{multline*}
N_{\eta^*}\K = \Big\{ c \, U_1V_1^T + c \, 
\widetilde{U} \widetilde{D} \widetilde{V}^T : 
\widetilde{U}^T \widetilde{U} = I, \;
\widetilde{V}^T \widetilde{V} = I, \;
\widetilde{D} = \mathrm{diag}(\widetilde{d}_1, \widetilde{d}_2,
\ldots), \\
U_1^T \widetilde{U} = 0, \; V_1^T \widetilde{V}=0, \;
\max(\widetilde{D}) \leq 1, \; c \geq 0\Big\},
\end{multline*}
and so the tangent space $T_{\eta^*}\K = (N_{\eta^*}\K)^\perp$ is
\begin{equation*}
T_{\eta^*}\K = \spa\Big(\Big\{U_1V_j^T, j=2,\ldots p\Big\} \cup
  \Big\{U_j V_1^T, j=2,\ldots n\Big\} \Big).
\end{equation*}
From this tangent space, the marginal variance $\sigma_{\eta^*}^2$ in
\eqref{eq:sigma} can be easily computed.  This leaves
\smash{$\V^-_{\eta^*},\V^+_{\eta^*},G_{\eta^*},H_{\eta^*}$} to be
addressed. 
As always, the quantities \smash{$\V^-_{\eta^*},\V^+_{\eta^*}$} can be 
determined numerically, as the optimal values of two convex
programs, see Section \ref{sec:linfrac}.  We now discuss
computation of the $(n+p-2)\times(n+p-2)$ matrices
\smash{$G_{\eta^*},H_{\eta^*}$}, and refer the reader 
to \citet{candes:svt,takemura:kuriki,takemura:kuriki1} for some of  
the calculations that follow.
The entries of $G_{\eta^*}$ are given by
\begin{align*}
G_{\eta^*}(U_1V_{i}^T, U_jV_1^T) &= G_{\eta^*}(U_jV_1^T, U_1V_j^T) \\  
&= \tr(V_jU_i^T C_{X,\Sigma}(\eta^*)), \\
G_{\eta^*}(U_1V_i^T, U_1V_j^T) &= G_{\eta^*}(U_iV_1^T,U_jV_1^T) \\ 
&= \delta_{ij} \cdot \tr(V_1U_1^T C_{X,\Sigma}(\eta^*)), 
\end{align*}
with $C_{X,\Sigma}(\eta^*)$ as in \eqref{eq:c}, which has a similar
computational form to $\sigma^2_{\eta^*}$, and can be computed
from knowledge of $T_{\eta^*}\K$ above.  The entries of $H_{\eta^*}$
are given by 
\begin{align*}
H_{\eta^*}(U_1V_{i}^T, U_jV_1^T) &= H_{\eta^*}(U_iV_1^T, U_1V_j^T) \\ 
&= \tr\big(V_iU_j^T X^T(y)\big) - d_1 \cdot 
G_{\eta^*}(U_jV_1^T, U_1V_i^T) \\ 
&= \delta_{ij} \cdot d_i - d_1 \cdot 
G_{\eta^*}(U_jV_1^T, U_1 V_i^T), \\ 
H_{\eta^*}(U_1V_{i}^T, U_1V_j^T) &= 
H_{\eta^*}(U_iV_{1}^T, U_jV_i^T) \\ 
&= \delta_{ij} \big[d_1 - \tr\big(V_1U_1^T
C_{X,\Sigma}(\eta^*)\big)\big].  
\end{align*}
When expressed in a suitable ordering of the above basis of the
tangent space, the matrix form of the above expressions are,
abbreviating $C=C_{X,\Sigma}(\eta^*)$,
\begin{align*}
G_{\eta^*} &= \begin{pmatrix}
U_1^T CV_1 \cdot I_{(n-1) \times (n-1)} &  
U_{-1}^TCV_{-1} \\
V_{-1}^TC^TU_{-1}  & U_1^T CV_1
\cdot I_{(p-1) \times (p-1)} 
\end{pmatrix}, \\
H_{\eta^*} &= \begin{pmatrix} \left(d_1 - U_1^T
    CV_1\right)  I_{(n-1) \times (n-1)} &  
D_{-1} - d_1 U_{-1}^TCV_{-1} \\
D_{-1}^T - d_1 V_{-1}^TC^TU_{-1} & \left(d_1 - U_1^T
  CV_1\right) 
 I_{(p-1) \times (p-1)}
\end{pmatrix},
\end{align*}
where $U_{-1}$ denotes all but the first column of $U$, with $V_{-1}$
similarly defined, and $D_{-1}$ denotes the matrix 
$\text{diag}(d_2, d_3 ,\ldots)$ with zeros added below in such a way
that $D_{-1}$ is an $(n-1) \times (p-1)$ matrix.  That is, if
$r=\mathrm{rank}(X^T(y))$,
\begin{equation*}
D_{-1} = 
\begin{pmatrix}
\text{diag}(d_2, \ldots d_r) \\
0_{(n-r) \times (p-1)}
\end{pmatrix}.
\end{equation*}
In the special case of principal components analysis, in which $X$ is 
the identity map on $\real^{n \times p}$, the above expressions
simplify to $G_{\eta^*}=I_{(n+p-2) \times (n+p-2)}$ and 
\begin{equation*}
H_{\eta^*} = \begin{pmatrix}
0 & D_{-1} \\
 D_{-1}^T & 0
\end{pmatrix},
\end{equation*}
and the $(n+p-2)$ eigenvalues of $\Lambda_{\eta^*} = G_{\eta^*}^{-1}
H_{\eta^*}$ are seen to be 
\begin{equation*}
\{\pm d_j, j=2,\ldots r\} \cup [0] \cdot (n+p - 2r). 
\end{equation*}

\begin{figure}
\begin{center}
\subfigure[Our test]{
\label{fig:nuclearpvala}
\includegraphics[width=0.5\textwidth]{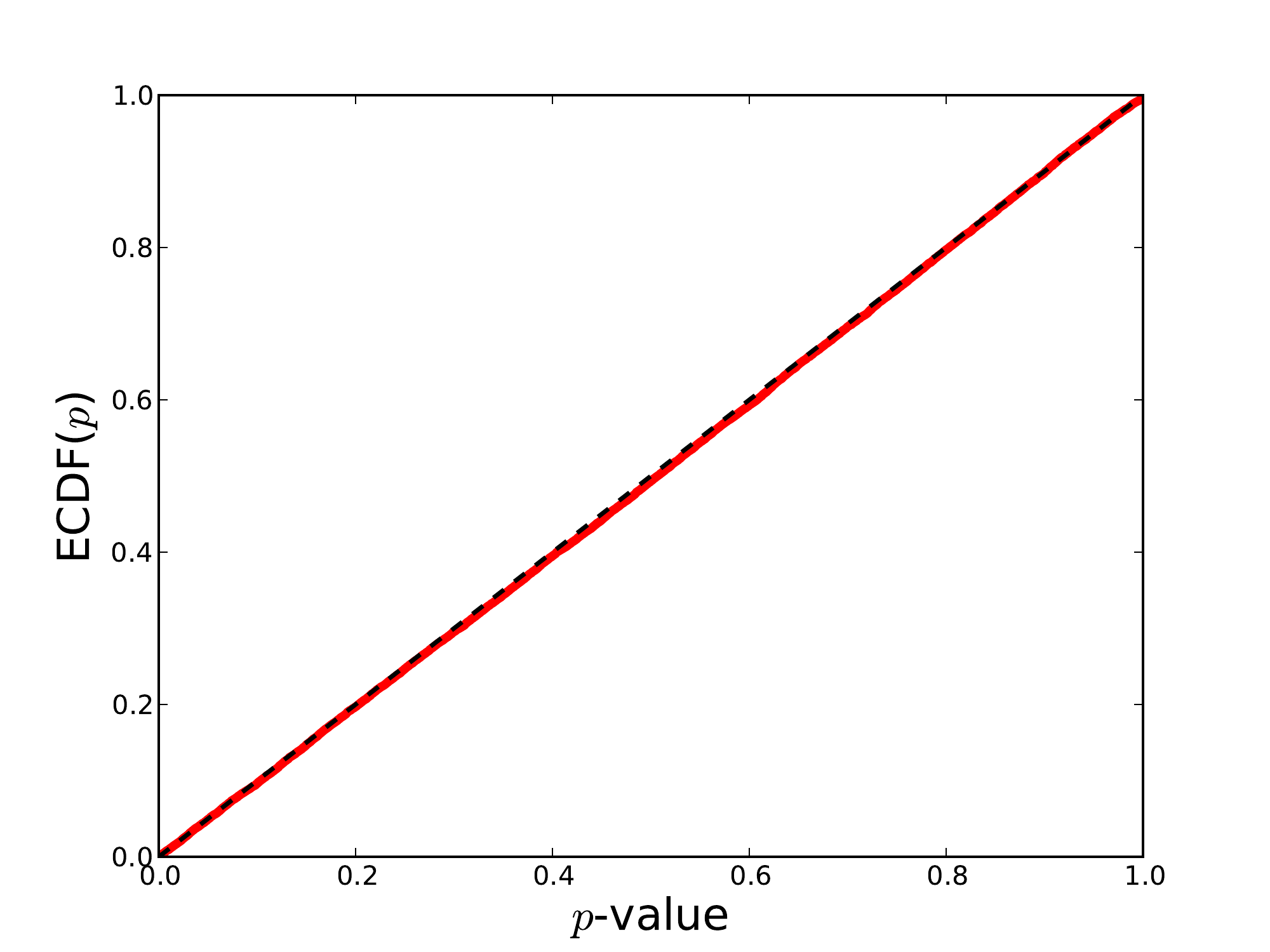}}
\hspace{-15pt}
\subfigure[Covariance test]{
\label{fig:nuclearpvalb}
\includegraphics[width=0.5\textwidth]{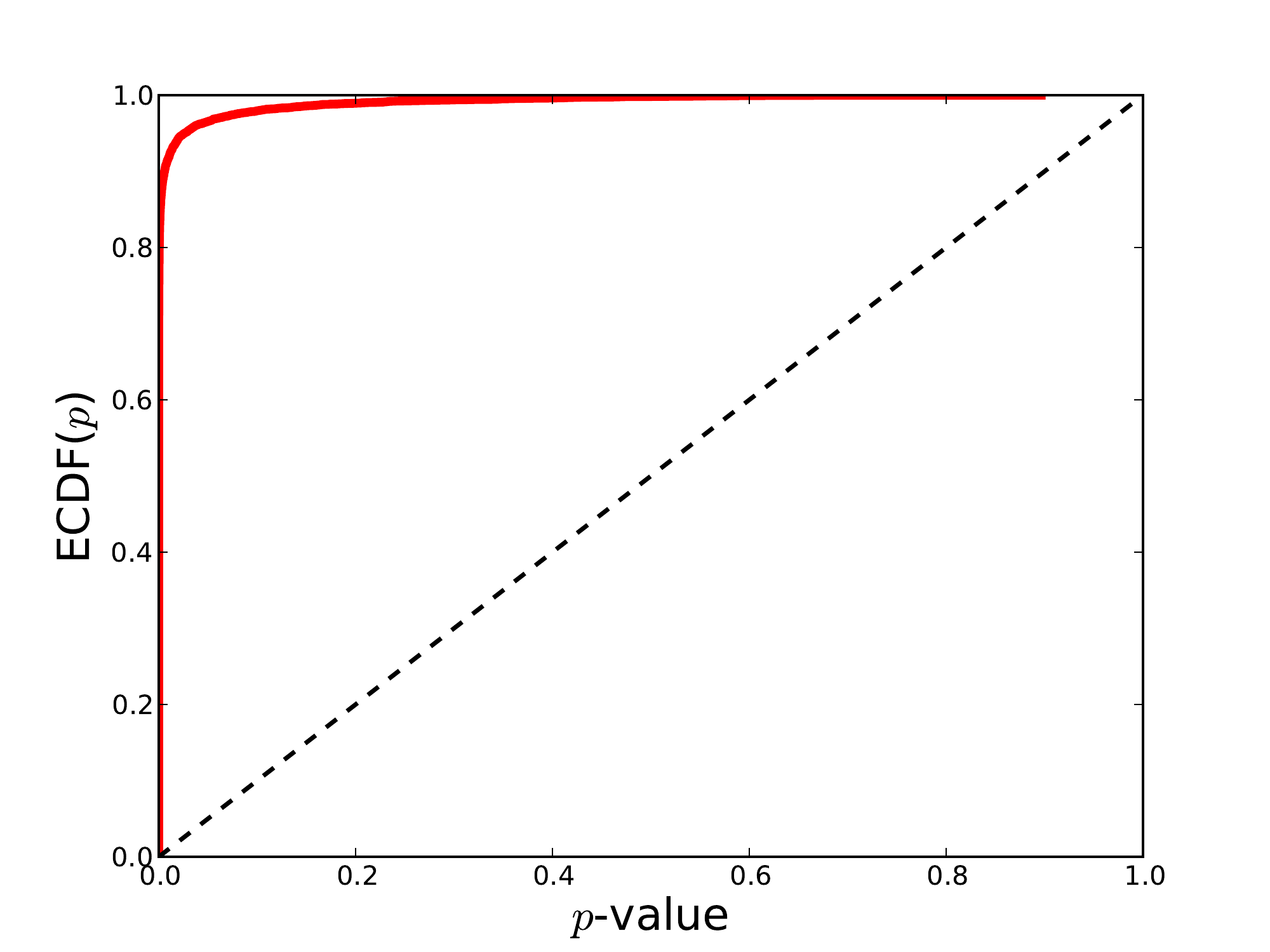}}
 \end{center}
\caption{\small\it The left panel shows the empirical distribution
  function of a sample of 20,000 p-values computed over a variety of
  problem setups that utilize the nuclear norm penalty. The right
  panel shows the distribution of 20,000 covariance test
  p-values when an $\mathrm{Exp}(1)$ approximation is used, which
  shows the exponential approximation to be clearly inappropriate for
  the nuclear norm setting.}
\label{fig:nuclear:pval}
\end{figure}

In Figure \ref{fig:nuclear:pval}, we plot the empirical distribution
function of a sample of 20,000 p-values computed over problem
instances that have been randomly sampled from the following 
scenarios, all employing the nuclear norm penalty (and all under 
the null model $B_0=0$, with $B_0$ being the underlying coefficient
matrix): 
\begin{itemize}
\item {\it principal components analysis:} 
$y$ is $2 \times 2, 3 \times 4, 50 \times 50, 100 \times
20, 30 \times 1000, 30 \times 5, 1000 \times 1000$; 
\item {\it matrix completion:} $y$ is 
$10 \times 5$ with 50\% of its entries observed at random,
$100 \times 30$ with 20\% of its entries observed at random,
$10 \times 5$ with a nonrandom pattern of observed entries,
$20 \times 10$ with a nonrandom pattern of observed entries,
$200 \times 10$ with 10\% of its entries observed at random; 
\item {\it reduced rank regression:} $X$ is $100 \times 10$ whose
  entries are drawn from the compound symmetric Gaussian distribution 
  with correlation 0.5, and $y$ is $100 \times 5$.
\end{itemize}
As is evident in the plot, the agreement with uniform 
is excellent (as above, each particular scenario above produces
$\text{Unif}(0,1)$ $p$-values regardless of how $X$ was chosen, modulo
the Morse assumption).

Again, along the lines of the covariance test, we consider
approximation of
\smash{$\lambda_1(\lambda_1-\V^-_{\eta^*})/\sigma^2_{\eta^*}$} by an
$\mathrm{Exp}(1)$ distribution under the null hypothesis.  Figure
\ref{fig:nuclearpvalb} shows that this approximation is quite far off,
certainly much more so than in the other examples.  Preliminary
calculations confirm mathematically that the $\mathrm{Exp}(1)$
distribution is not the right limiting distribution here; we will
pursue this in future work. 


\section{Finding $\V^-_{\eta}, \V^+_{\eta}$}
\label{sec:linfrac}

Earlier, in Remark \label{rem:convexity} we noted that convexity of $\K$ was not needed for much of the work in this paper. It is only when we wish to actually
compute the test statistic that we restrict to convex $\K$. This section describes
convex optimization problems that can be solved to find the values of $\V^{\pm}_{\eta}$ which
are one of  the key stumbling pieces to computing the test statistic. Nevertheless,
if we had access to a computer that could compute these quantities, along with
$\Lambda_{\eta}$ we would
not need to restrict interest to convex $\K$. We are still assuming positive reach as this is an
important assumption to ensure the modified processes $\tf^{\eta}_{\eta}$ are not 
singular. See Chapter 14 of \cite{RFG} for the case when $\K$ is a smooth locally convex (i.e. positive reach) subset of the
Euclidean sphere. This paper considers arbitrary smooth subsets of positive reach
up to this point.

Inspecting Lemma \ref{lem:mainresult}, we see that in order to compute the $p$-value in \eqref{eq:sunif1}, we must find
$\V^-_{\eta}, \V^+_{\eta}$ at $\eta=\eta^*$.
Recall the definition of $\V^-_{\eta}$ in \eqref{eq:vplus},
\begin{equation*}
\V^-_{\eta} = \max_{z \in \K: \, z^T C_{X,\Sigma}(\eta) < 1} \,
\frac{\tf^{\eta}_z - z^T C_{X,\Sigma}(\eta) \cdot \tf^{\eta}_{\eta}}
{1 - z^T C_{X,\Sigma}(\eta)},
\end{equation*}
where $C_{X,\Sigma}(\eta)$ is defined as in \eqref{eq:c}, and can be
expressed more concisely as
\begin{equation*}
z^T C_{X,\Sigma}(\eta) = 
\frac{z^TX^T(I-P_{\eta,X,\Sigma})\Sigma X \eta}
{\eta^TX^T(I-P_{\eta,X,\Sigma}) \Sigma X\eta}.
\end{equation*}
Recall that at $\eta=\eta^*$ we have
\smash{$\tf^{\eta}_z=f_z$} for all $z$, and
\smash{$\tf^{\eta}_{\eta}=\lambda_1$},
so 
\begin{equation*}
\tf^{\eta^*}_z - z^T C_{X,\Sigma}(\eta^*) \cdot \tf^{\eta^*}_{\eta^*} =  
z^T \big(X^T y - C_{X,\Sigma}(\eta^*) \cdot \lambda_1 \big).
\end{equation*}
Therefore, $\V^-_{\eta^*}$ is given by the maximization
problem
\begin{align}
\nonumber
\V^-_{\eta^*} = \; & \max_{z \in \real^p} \qquad\quad
\frac{z^T \big(X^Ty - C_{X,\Sigma}(\eta^*) \lambda_1\big)}
{1 - z^T C_{X,\Sigma}(\eta^*)} \\
\label{eq:linfrac}
& \st \quad 
z \in \K, \; 1 - z^T C_{X,\Sigma}(\eta^*) > 0. 
\end{align}
This is a generalized {\it linear-fractional} problem \citep{boyd}.
(We say ``generalized'' here because the constraint $z \in \K 
\Longleftrightarrow \pen(z) \leq 1$ can be seen as an infinite
number of linear constraints, while the standard linear-fractional 
setup features a finite number of linear constraints.)  Though 
not convex, problem \eqref{eq:linfrac} is quasilinear (i.e., both
quasiconvex and quasiconcave), and further, it is equivalent to a
convex problem.  Specifically, 
\begin{align}
\nonumber
\V^-_{\eta^*} = \; & \max_{u \in \real^p, \, w \in \real}
\quad\; u^T \big(X^Ty - C_{X,\Sigma}(\eta^*) \lambda_1\big) \\
\label{eq:vpluscvx}
& \st \quad 
\pen(u) \leq w, \; w - u^T C_{X,\Sigma}(\eta^*) = 1, \;
w \geq 0.
\end{align}
The proof of equivalence between \eqref{eq:linfrac} and
\eqref{eq:vpluscvx} follows closely Section 4.3.2 of
\citet{boyd}, and so we omit it here. Similarly, $\V^+_{\eta^*}$ is
given by the convex minimization problem   
\begin{align}
\nonumber
\V^+_{\eta^*} = \; & \min_{u \in \real^p, \, w \in \real}
\quad\; u^T \big(X^Ty - C_{X,\Sigma}(\eta^*) \lambda_1\big) \\
\label{eq:vminuscvx}
& \st \quad 
\pen(u) \leq w, \; w - u^T C_{X,\Sigma}(\eta^*) = -1, \;
w \geq 0.
\end{align}

In the worst case, $\V^-_{\eta^*}$ and $\V^+_{\eta^*}$ can be
determined numerically by solving the problems \eqref{eq:vpluscvx} 
and \eqref{eq:vminuscvx}.  Depending on the penalty $\pen$, one may
favor a particular convex optimization routine over another for this
task, but a general purpose solver like ADMM \citep{admm} should be
suitable for a wide variety of penalties.
In several special cases involving the lasso, group lasso, and
nuclear norm penalties, $\V^-_{\eta^*}$ and $\V^+_{\eta^*}$
have closed-form expressions.  We state these next, but in the
interest of space, withhold derivation details. We leave more precise  
calculations to future work.  

\subsection{Lasso and group lasso}

The lasso problem is a special case of the group lasso where 
all groups have size one; therefore the formulae derived here for
the group lasso also apply to the lasso.  (In the lasso case,
actually, the angles $\psi^{\pm}$ below are always $\pm \pi$.) 

For any group $g\in\cG$, let $\theta(\eta,g)$ be the
angle between $X_g^Ty - C_{X,\Sigma}(\eta)_g$ and
$C_{X,\Sigma}(\eta)_g$ [where recall that $C_{X,\Sigma}(\eta)$ is as
defined in \eqref{eq:c}], and define the angles $\psi^{\pm}(\eta, g)$ by   
\begin{equation*}
\sin \psi^{\pm}(\eta, g) = \frac{\|C_{X,\Sigma}(\eta)_g\|_2}{w_g}
\cdot \sin \theta(\eta,g).
\end{equation*}
Define the quantities
\begin{equation*}
w^{\pm}(\eta,g) = \frac{\|X_g^Ty-C_{X,\Sigma}(\eta)_g\|_2 \cdot
\cos \psi^{\pm}(\eta,g) }{w_g - \|C_{X,\Sigma}(\eta)_g\|_2 \cdot
\cos\big(\theta(\eta,g)-\psi^{\pm}(\eta,g)\big)},
\end{equation*}
and
\begin{align*}
v^+(\eta,g) &= \begin{cases}
\min\{w^{\pm}(\eta,g)\} & \text{if $\|C_{X,\Sigma}(\eta)_g\|_2 
\geq w_g$} \\ 
\max\{w^{\pm}(\eta,g)\} & \text{otherwise}
\end{cases}, \\
v^-(\eta,g) &= \begin{cases}
\max\{w^{\pm}(\eta,g)\} & \text{if $\|C_{X,\Sigma}(\eta)_g\|_2 
\geq w_g$} \\ 
\infty & \text{otherwise}
\end{cases}.
\end{align*}
Then $\V^-_{\eta^*},\V^+_{\eta^*}$ have the explicitly computable form  
\begin{align*}
\V^-_{\eta^*} &= \max_{g \neq g^*} \, v^+(\eta^*,g), \\
\V^+_{\eta^*} &= \min_{g \neq g^*} \, v^-(\eta^*,g).
\end{align*}

\subsection{Nuclear norm}

For principal components analysis, when
$X=I$, it is not difficult
to show that  $\V^-_{\eta^*} = d_2=-\min(\Lambda_{\eta^*})$ and
$\V^+_{\eta^*}=\infty$.    
Numerically, this seems to be true even for an arbitrary $X$ (reduced
rank regression) or arbitrary patterns of missingness (matrix
completion). 
We remark that computing
$\V^-_{\eta^*}=d_2$ requires solving an eigenvalue problem of size at
least $\max(n,p) \times \max(n,p)$.  (The ADMM approach for solving
problem \eqref{eq:vpluscvx} is no better, as each iteration involves 
projecting onto the nuclear norm epigraph, which requires an 
eigendecomposition.)  This is quite an expensive computation, and a
fast approximation is desirable. We leave this for future work. 

\begin{figure}\begin{center}
\subfigure[Group lasso]{
\includegraphics[width=0.5\textwidth]{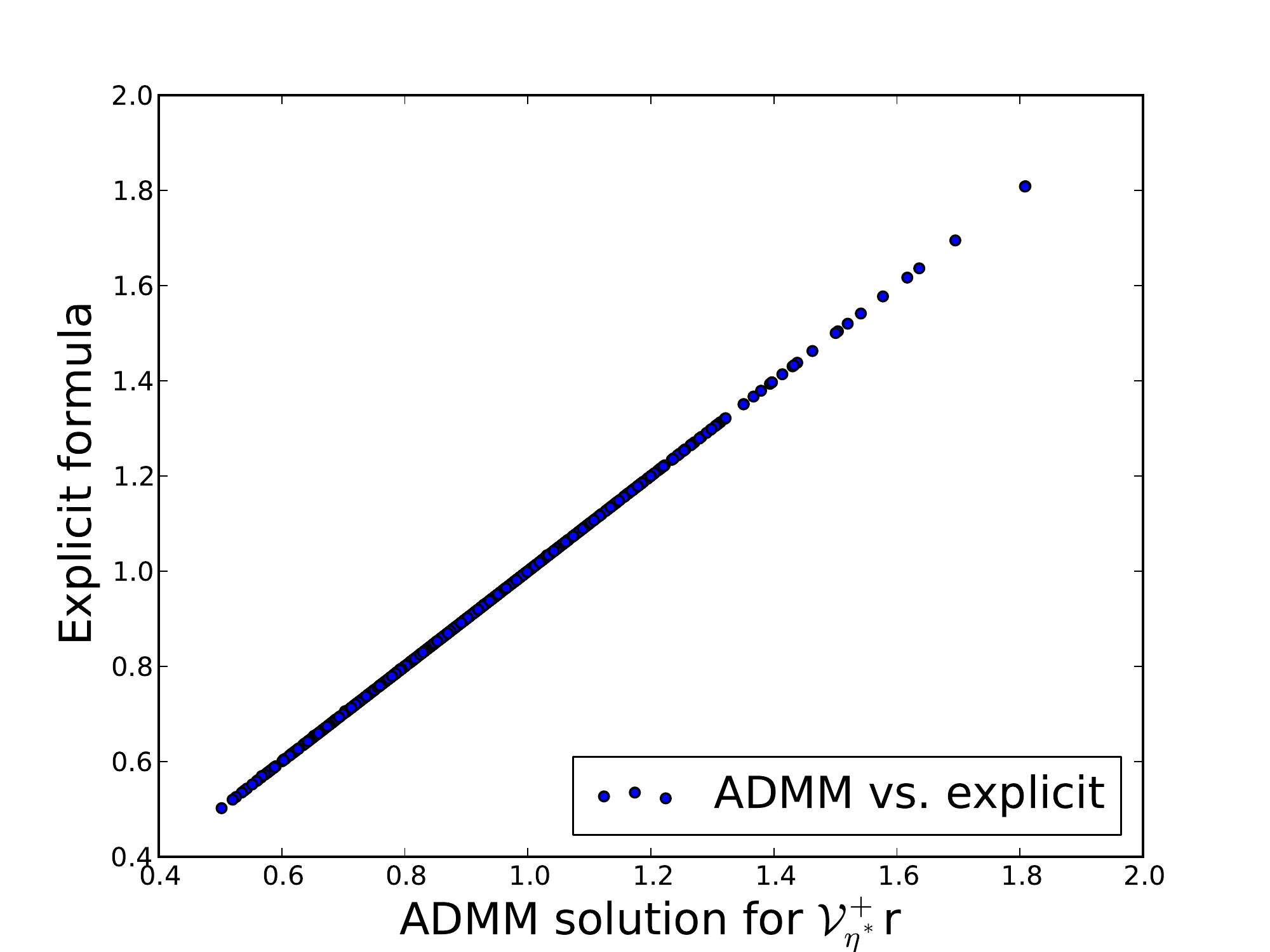}}
\hspace{-15pt}
\subfigure[Nuclear norm]{
\includegraphics[width=0.5\textwidth]{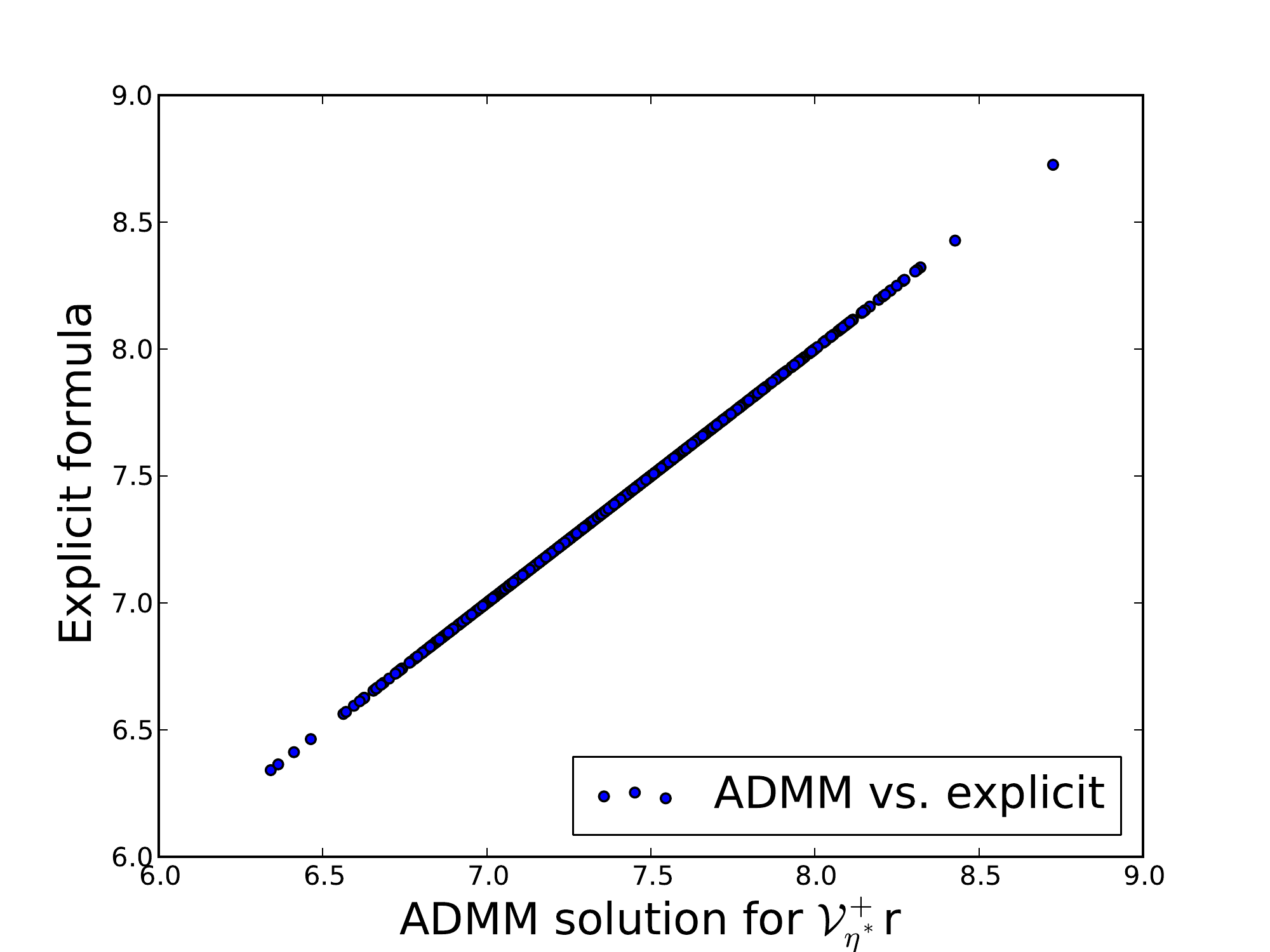}}
 \end{center}
\caption{\small\it Comparison of the explicit (analytically derived)
  form of $\V^-_{\eta^*}$ versus the solution of convex program
  \eqref{eq:vpluscvx} found numerically by ADMM, for the group 
  lasso and nuclear norm penalties.} 
\label{fig:comparison}
\end{figure}

\subsection{Relation to second knot in the solution path}

\begin{figure}
\begin{center}
\subfigure[Group lasso]{
\label{fig:knota}
\includegraphics[width=0.5\textwidth]{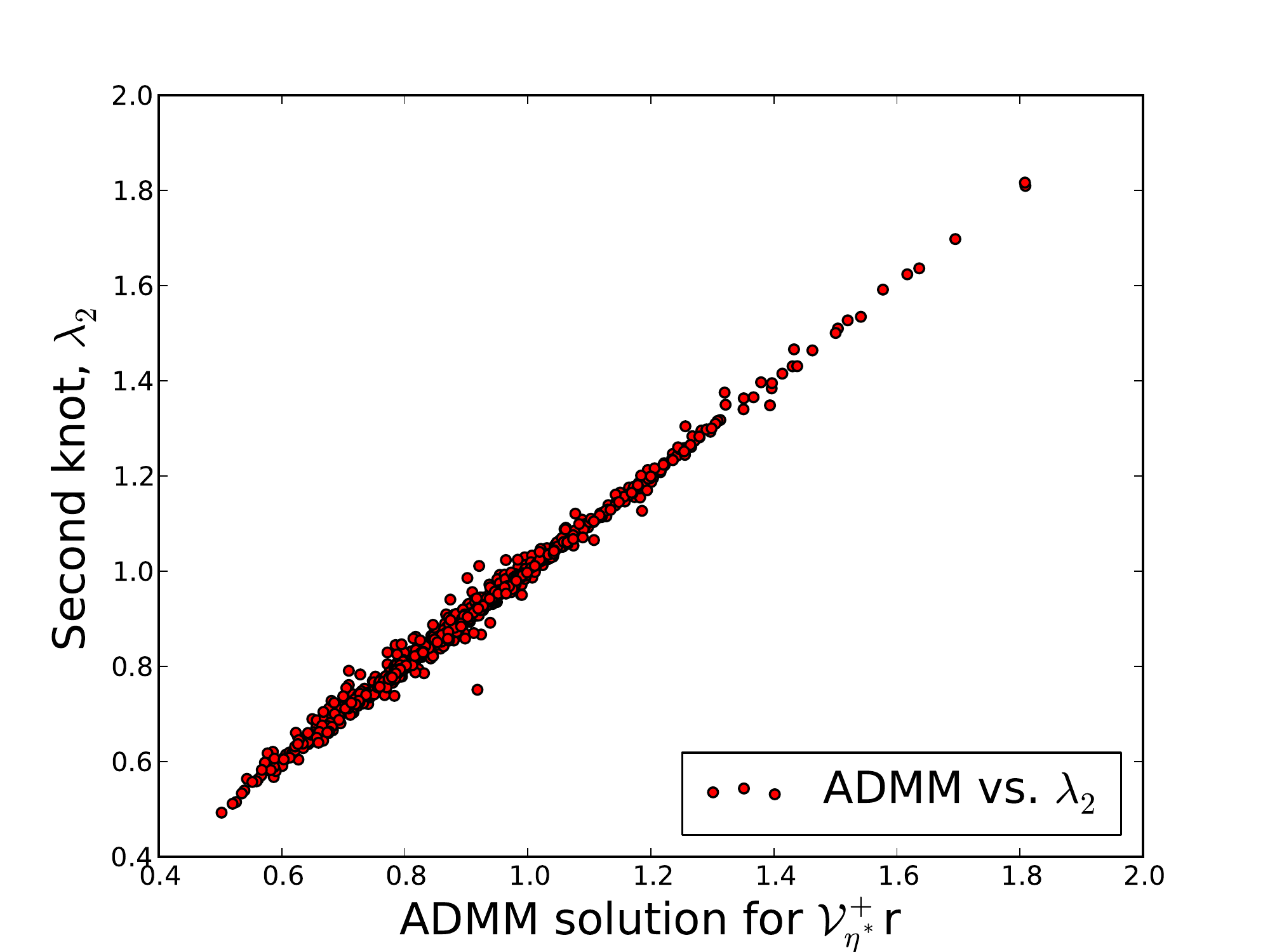}}
\hspace{-15pt}
\subfigure[Nuclear norm]{
\label{fig:knotb}
\includegraphics[width=0.5\textwidth]{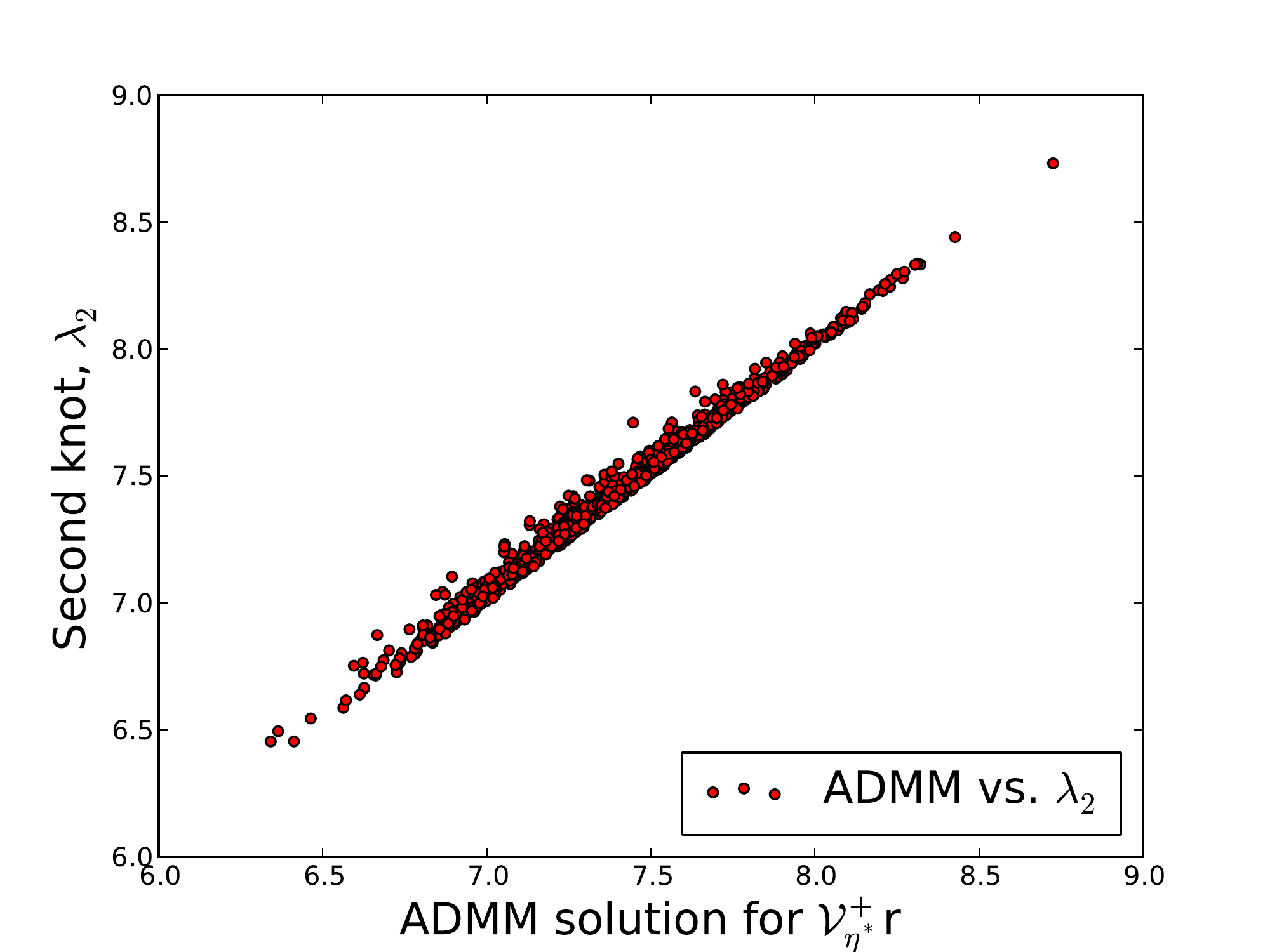}}
\end{center}
\caption{\small \it Comparison of $\V^-_{\eta^*}$ with the second knot 
  $\lambda_2$ in the solution path, for the group lasso and the
  nuclear norm problems. For the group lasso, $\lambda_2$ is defined
  to be the largest value of $\lambda$ such that only one group of
  variables is active. For the nuclear norm, it is defined as the
  largest value of $\lambda$ for which the solution has rank 1.}  
\end{figure}

The quantity $\V^-_{\eta^*}$ is always a lower bound on the first knot
$\lambda_1$, and possesses an interesting connection to another tuning
parameter value of interest along the solution path.  In particular,
$\V^-_{\eta^*}$ is related to the second knot $\lambda_2$, which can be
interpreted more precisely in a problem-specific context, as follows:
\begin{itemize}
\item for the lasso, with standardized predictors,
  $\|X_j\|_2$ for all $j=1,\ldots p$, we have exactly
  $\V^-_{\eta^*}=\lambda_2$, the value of the tuning parameter
  $\lambda$ at which the solution changes from one nonzero component 
  to two [this follows from a calculation as in Section 4.1
  of \citet{covtest}];
\item for the group lasso, as shown in Figure \ref{fig:knota},
  $\V^-_{\eta^*}$ is numerically very close to the value $\lambda_2$
  of the parameter at which the solution changes from one nonzero
  group of components to two;
\item for the nuclear norm penalty, as shown in Figure
  \ref{fig:knotb}, 
$\V^-_{\eta^*}$ is numerically very close to the value $\lambda_2$ of
the parameter at which the solution changes from rank 1 to rank 2.
\end{itemize}

\section{Non-Gaussian errors}
\label{sec:nongaussian}

Throughout, our calculations have rather explicitly used the fact that 
$X^Ty$ is Gaussian distributed. One could potentially appeal to
the central limit theorem if the components of $y-X\beta_0$ are
i.i.d. from some error distribution (treating $X$ as fixed), though
the calculations in this work focus on 
the extreme values, so the accuracy of the central limit theorem in
the tails may be in doubt. In the interest of space, we do not address 
this issue here. 

Instead, we consider a scenario with heavier
tailed, skewed noise and present simulation results.  In particular,
we drew errors according to a $t$-distribution with 5
degrees of freedom plus an independent centered $\mathrm{Exp}(1)$.  
Figure \ref{fig:nongauss} shows that the lasso and group lasso
p-values are relatively well-behaved, while the nuclear norm p-values
seem to break down.

\begin{figure}
\begin{center}
\subfigure[Lasso]{
\includegraphics[width=0.6\textwidth]{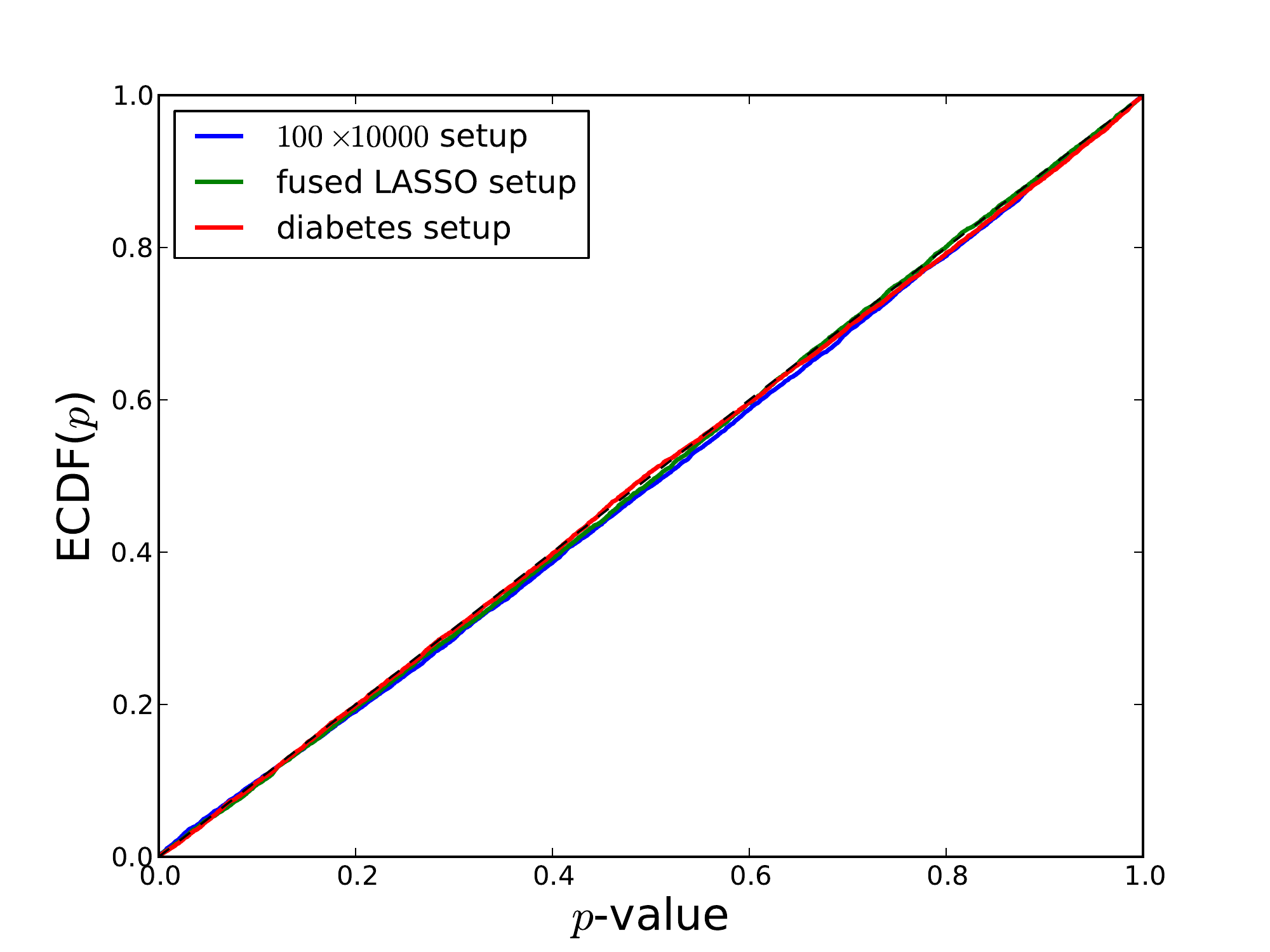}}
\subfigure[Group lasso]{
\includegraphics[width=0.6\textwidth]{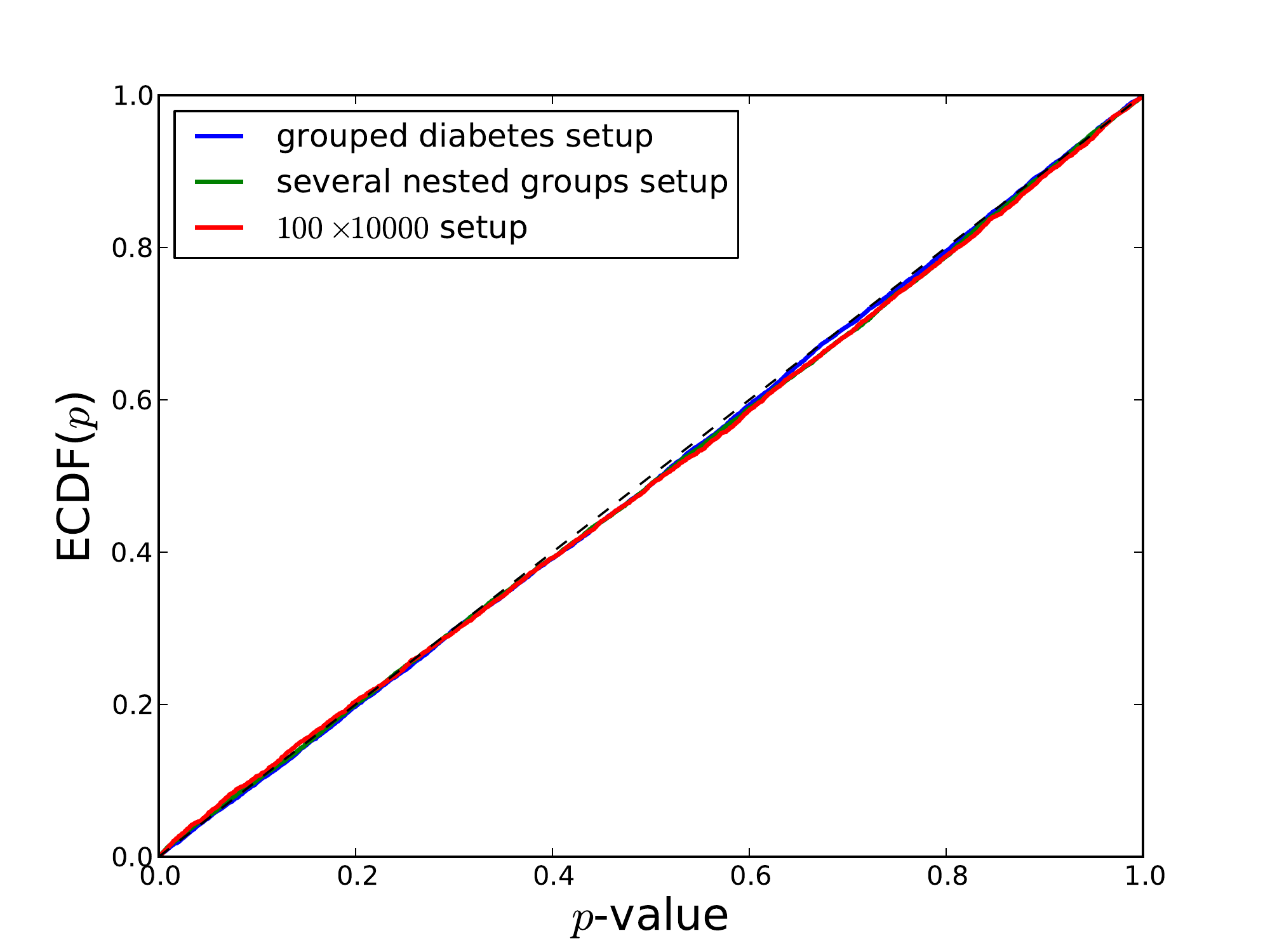}}
\subfigure[Matrix completion]{
\includegraphics[width=0.6\textwidth]{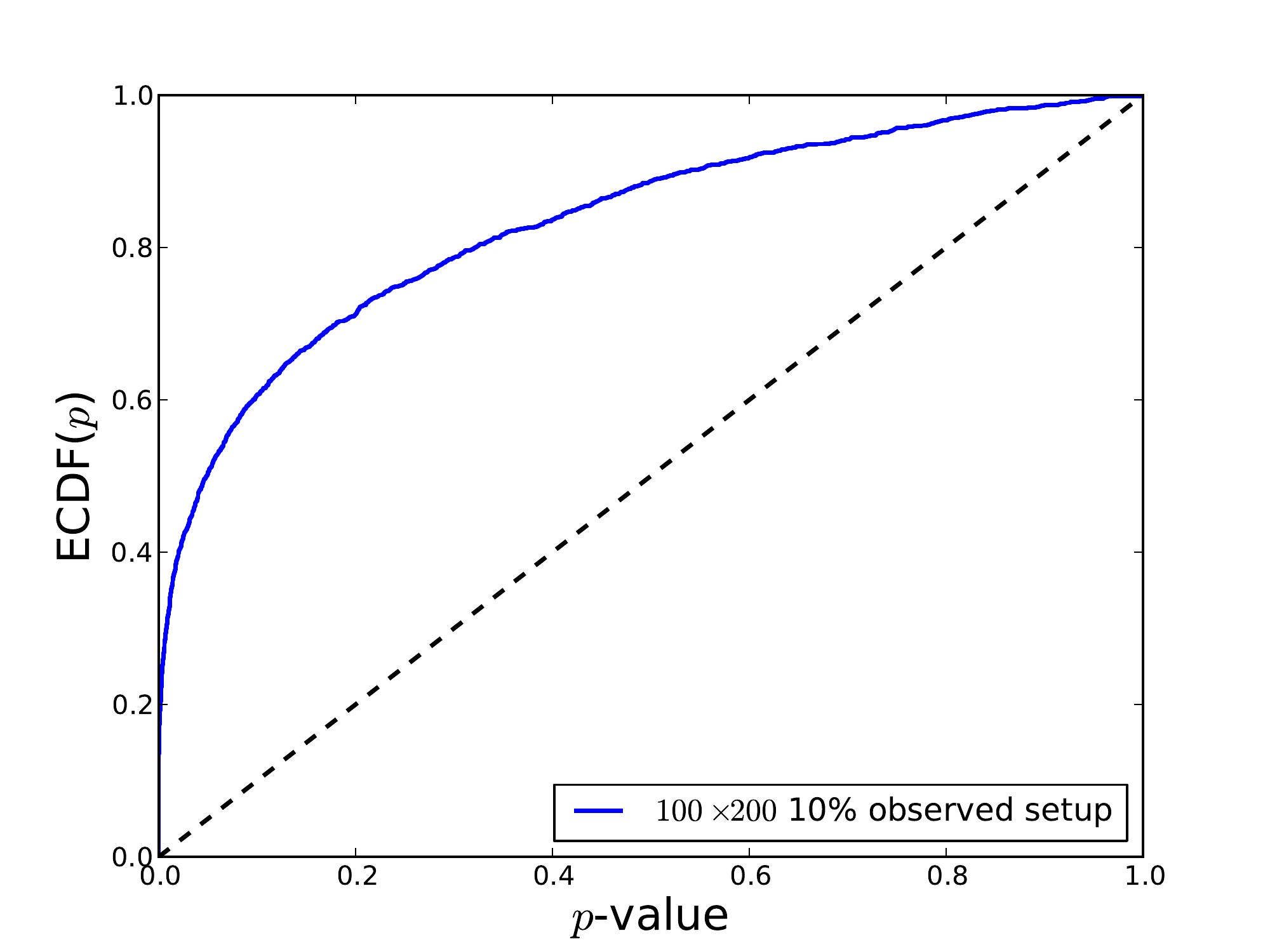}}
\end{center}
\label{fig:nongauss}
\caption{\small \it Empirical distribution function of p-values
  for the lasso, group lasso and matrix completion problems under 
heavy-tailed, skewed noise.}
\end{figure}

\section{Discussion}
\label{sec:discussion}

We derived an exact (non-asymptotic) p-value for 
testing a global null hypothesis in a general regularized regression
setting.  Our test is based on a geometric characterization
of regularized regression estimators and the Kac-Rice formula, and has
a close relationship to the covariance test for the lasso problem
\citep{covtest}.  In fact, the $\mathrm{Exp}(1)$ limiting null
distribution of the covariance test can be derived from the formulae
given here.  These two tests give similar results for lasso problems,
but our new test has exact (not asymptotic) error control under the
null, with less assumptions on the predictor matrix $X$.  

Another strength of our approach is that it provably extends well
beyond the lasso problem; in this paper we examine tests for both the
group lasso and nuclear norm regularization problems.  Still, the
test can be applied outside of these settings too, and is limited only
by the difficulty in evaluating the p-value in practice (which relies
on geometric quantities to be computed).
 
We recall that the covariance test for the lasso can be
applied at any knot along the solution path, to test if the
coefficients of the predictors not yet in the current model are all
zero.  In other words, it can be used to test more refined null
hypotheses, not only the global null.
We leave the extension of Theorem \ref{thm:mainresult} to this more 
general testing problem for future work.  

\subsection{Power}

One important issue we have not yet addressed is the power of our
tests. While the setting under which our tests are applicable is quite
broad (in the lasso problem, the only assumption needed is that $X$
have columns in general position), an end user will likely be
interested in how powerful our test is. 

\subsubsection{Lasso}

For the lasso case, our test statistic is based on a conditional
distribution of $\|X^Ty\|_{\infty}$ [referred to as the Max test in
\cite{global_testing}]. Therefore, the best 
one can expect is to have similar results to the Max test in this situation. 
For general $X$, a simple sufficient condition for full asymptotic power is that the gap 
$\lambda_1 - \lambda_2$ must decrease to 0 no faster than $\lambda_1 \to \infty$.  
This follows from the Mills' ratio approximation 
$$
\frac{1 - \Phi(Z_{(1)})}{1 -\Phi( Z_{(2)})} \approx \exp\left(-\frac{1}{2}(\lambda_1-\lambda_2)(\lambda_1 + \lambda_2) \right).
$$

We note here that, in principle, by taking sequences of problems $(X_n,y_n)$ the set of limiting distributions for $\|X_n^Ty_n\|_{\infty}$ under $H_0$ 
is effectively the set of all possible distributions for the maximum absolute value of a 
(separable) centered
Gaussian process. This is a huge class, implying that the
study of power for such a problem is also a very broad problem. The Kac-Rice test
is based on a conditional distribution for the Max test statistic and is applicable for 
almost all matrices $X$.

The behavior of the Max test has been established for low coherence designs  in \cite{global_testing},
though we emphasize that our test makes no assumption about
coherence. In the interest of space, we 
consider the power of our test to the Max test in
the orthonormal design case. While this is a stronger assumption than low
coherence, we expect a similar situation to hold in the low coherence
setting discussed in . Attempting to prove that
these orthogonal results carry over to a low coherence scenario is an
interesting problem for future research and is beyond the scope of
this paper.  

In the orthonormal case, our test statistic reduces in distribution to 
$$
\frac{1 - \Phi(Z_{(1)})}{1 - \Phi(Z_{(2)})},
$$
where $Z_i \sim |N(\mu_i,1)|$ independently for $i=1,\ldots n$, and
$Z_{(1)} \geq Z_{(2)} \geq \ldots \geq Z_{(n)}$ are their order
statistics in decreasing order.

We first consider the case of finite sparsity and identity design.
The simplest possible alternative hypothesis
is the 1-sparse case where a single $\mu_i$ is nonzero. Let
$\mu_1 = r \sqrt{2\log(p)}$ and all other $\mu_i = 0$. If $r$ is some
constant greater than 1, then with high probability the first
knot $\lambda_1$ will be achieved by $Z_1$. 
Standard Gaussian tail bounds can be used to upper bound the 
Kac-Rice pivot and we see that for any $r > 1$ the upper bound goes to 0 as $p \to \infty$, so
in this case the test has asymptotic full power at the same threshold
as Bonferroni. This is the best possible threshold for asymptotic power
against the 1-sparse alternative. Similar statements hold for $k$-sparse alternatives
where $k$ is considered fixed.

We now consider the sparsity growing with $n$.
Following \cite{global_testing}, we set $\lfloor n^{1-\delta} \rfloor$
of the underlying means $\mu_i$, $i=1,\ldots n$ to be nonzero and
equal to $a$. This is arguably the worst case for our test, which 
is based on the gap between $Z_{(1)}$ and $Z_{(2)}$. 
Let $A(\delta,n)= \sqrt{2 \log n} \cdot (1 - \sqrt{1-\delta})$; this 
is the threshold for non-trivial asymptotic for the Max test. That is,
for any $\epsilon > 0$, the Max test is asymptotically power powerful
(Type I + Type II error $\to 0$) if the nonzero means have value
$a = A(\delta,n) + \epsilon \sqrt{2 \log n}$, and asymptotically powerless
(Type I + Type II error $\to 1$ or more) if the nonzero means have
value $a = A(\delta,n) -  \epsilon \sqrt{2 \log n}$.  



A fairly straightforward argument based on the spacings
calculations in \cite{covtest}, the asymptotic power of the Kac Rice level $\alpha$ test is
$\alpha^{\sqrt{1-\delta}/(1+\epsilon)}$ and Type I + Type II error =
$\alpha + 1 - \alpha^{\sqrt{1-\delta}/(1+\epsilon)} < 1$. 
As $\delta \to 1$, we recover the full power in the finite sparsity case.
 
\subsubsection{Group lasso}

The authors are not aware of literature on the global power of 
tests such as the group lasso 
In practice, using the group lasso penalty allows the modeler freedom
to favor certain groups by judicious choices of the group
weights. Other interesting examples of the group lasso such as 
{\em glinternet} \citep{glinternet} generally have widely varying
group sizes. 

Nevertheless,
if we are willing to make simplifying assumptions on
group sizes and the designs, then we can say something about power. 
The simplest thing might be to again assume orthonormal design, with
the additional assumption of  
equal sized groups and weights. In this very specialized setting,
questions of power reduce to questions about 
order statistics of non-central $\chi_k$ random variables for a fixed
group size $k$. These random variables 
all have similar tail behavior to $N(0,1)$ random variables, with the effective number of observations
now being $n/k$. We therefore expect similar behavior to the Max test
if $k$ is fixed. Allowing $k$ to vary with $n$, or even be random 
seems an interesting problem to consider, even in the orthonormal
design setting. For general $X$ without coherence assumptions 
this seems like a challenging  problem indeed.

\subsubsection{Nuclear norm}

The authors are not aware of other approaches to testing the global
null using the test motivated by the nuclear norm, based 
the largest singular value of $X^Ty$. The special case $X=I$ is an
obvious exception, in which this largest singular 
value corresponds to the edge of the spectrum and much is known about 
the limiting distribution, the Tracy-Widom law \cite{johnstone}. 
In this case, our test statistic is based on the conditional
distribution of $\lambda_1$ given $\lambda_2, \dots,
\lambda_{\text{min}(n,p)}$. 
Rather than begin a detailed analysis of the power, we provide a small
simulation study to compare 
to existing implementations based on the Tracy-Widom limit. The simulation
is the work of Yunjin Choi, currently a Ph.D. student investigating
the use of the Kac-Rice pivot to inference in PCA. The example
is a rank-one example, demonstrating that the Kac-Rice test is competitive
with the Tracy-Widom approximation of \cite{johnstone}. The Kac-Rice test
has the advantage of control of Type I error and applicability to
the matrix completion or reduced-rank regression setting.

\begin{figure}
\begin{center}
\label{fig:TW}
\includegraphics[width=\textwidth]{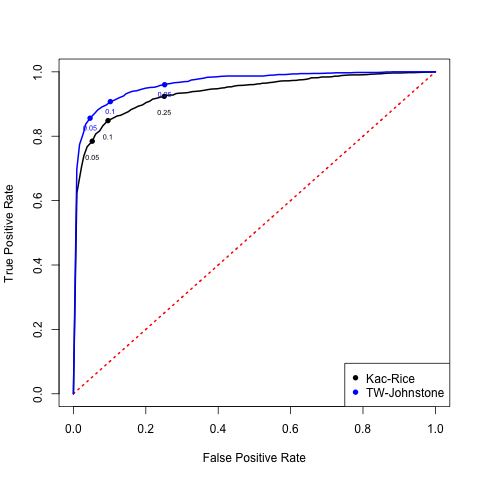}
\hspace{-15pt}
\caption{\small \it  Comparison of Kac-Rice test
to Tracy Widom test for $(n, p) = (50,10)$. The alternative
is a rank 1 matrix with singular value $\sqrt{50*10*0.125}=7.90$.
We see that the Tracy Widom test is more powerful here, probably because it conditions
on less then the Kac-Rice test, which conditions on all singular values of the data matrix.
See the remarks following Lemma \ref{lem:mainresult}. The Kac-Rice test
is applicable to situations beyond PCA such as reduced-rank regression and
matrix completion while the Tracy-Widom approximation is not, to the
authors' knowledge.} 
\end{center}
\end{figure}

\subsection{Related work on selective inference}
\label{sec:selectiveinf}

We conclude this paper with a discussion of how results in this paper can be applied to
{\em selective inference} and how this paper relates to other recent work.
 An astute reader will note that all the distributional
results in Section \ref{sec:kacrice} are valid for {\em any $\mu$} and not just the global null
hypothesis $H_0:\mu=0$. In this sense, the results in this paper go beyond just hypothesis tests
as they can be used to construct intervals containing linear functions of the true mean $\mu$. 
Formally, they can be used to construct intervals for $\mu_{\eta^*}$ defined in \eqref{eq:mu}.
Specifically, consider the set
\begin{equation}
\label{eq:selection_interval}
SI = \left\{\delta: \text{min}\left(
    \Ss_{\Lambda_{\eta^*},\V_{\eta^*},\delta,\sigma^2_{\eta^*}}(f_{\eta^*})
    , 1-
    \Ss_{\Lambda_{\eta^*},\V_{\eta^*},\delta,\sigma^2_{\eta^*}}(f_{\eta^*})
  \right) > \alpha/2 \right\}. 
\end{equation}
where $\Ss$ is defined in \eqref{eq:surv}. As $\Ss$ is an exact pivot for $\mu_{\eta^*}$ we see that
\begin{equation}
\label{eq:selection_coverage}
\Pp \left(\mu_{\eta^*} \in SI \right) = 1 - \alpha.
\end{equation}

Applying the above to the first step of the lasso, we can construct
exact intervals that cover $X_{(1)}^T\mu$ where 
$X_{(1)}$ is the first variable chosen by the lasso.
In work initiated after the initial submission of this paper
\cite{exact_lasso}, one of the authors has used this selective
inference 
framework for exact selective inference for selected variables in the problem
\begin{equation}
\label{eq:lasso:fixed}
\minimize_{\beta} \frac{1}{2} \|y-X\beta\|^2 + \lambda \|\beta\|_1
\end{equation}
for some $\lambda > 0$ fixed. With a small modification, one can
similarly analyze the problem 
\begin{equation}
\label{eq:lasso:bound}
\minimize_{\beta: \|\beta\|_1 \leq \kappa } \frac{1}{2} \|y-X\beta\|^2 .
\end{equation}

 In the PCA setting, the above interval gives an interval that covers
a pseudo-singular value $U_1(y)^T\mu V_1(y)$. If $U_1(y), V_1(y)$
recovered $U_1(\mu), V_1(\mu)$ without error, then this would be an
actual singular value and the selection interval would be a confidence
interval. 
Under certain conditions for the mean matrix $\mu$, results in random
matrix theory \cite{debashis_johnstone} ensure that the 
random singular vectors converge in some sense to the population
singular vectors, in the sense that $U_1(\mu)^TU_1(y) \to 1$. 
Extending such results to the nuclear norm setting with $X \neq I$
seems like an interesting and challenging problem. 

Perhaps the strongest sense in which these results differ from earlier
results on selective inference is that they  
are exact and computationally feasible. For the lasso, for instance,
we need only assume that $X$ has columns in general position.

Another sense in which these results differ from existing results on
selective inference is  
that our testing problem allows for the possibility that the choice is
made from some continuous set. In the lasso setting, 
a ``mode'' is chosen by selecting a set of active variables and their
signs. This differs from the group lasso setting in the sense 
that the active subgradient varies continuously in some set. This
added complexity makes conditioning on the ``active  
subgradient'' a somewhat more technically dubious procedure. This
conditioning is handled explicitly via the  
Kac-Rice formula which, in this case, can be thought of as enabling us
to carry out a partition argument over a smooth partition set. 

Finally, the construction of the process $\tf^{\eta}_{\eta}$ is a
contribution to the theory of smooth Gaussian random fields as 
described in \cite{RFG}. Our construction allows
earlier proofs that apply only to (centered) Gaussian random fields of
constant variance (i.e. marginally stationary)  
to smooth random fields with arbitrary variance.  Even in the
marginally stationary case, 
the conditional distribution $\Qq$ defined in \eqref{eq:distq} 
provides a new tool for exact selective inference at critical points
of such random fields. We leave this, and many other topics, for
future work. 

\bigskip
\noindent

{\bf Acknowledgements:} We are deeply indebted to Robert Tibshirani, 
without whose help and suggestions this work would not have been
possible. Two referees and and an associate editor 
provided helpful feedback on the first version of this work that has
helped the authors improve this paper. 

\appendix
\section{Proofs and supplementary details}

\subsection{Uniqueness of the fitted values and penalty}
\label{app:uniqueness}

Consider the strictly convex function $f(u)=\|y-u\|_2^2/2$. Suppose
\smash{$\hbeta_1,\hbeta_2$} are two solutions of the regularized
problem \eqref{eq:genprob} with 
\smash{$X\hbeta_1\not=X\hbeta_2$}.  Note
\begin{equation*}
f(X\hbeta_1)+\lambda \pen(\hbeta_1)=f(X\hbeta_2)+\lambda
\pen(\hbeta_2) = c^*,
\end{equation*}
the minimum possible value of the criterion in \eqref{eq:genprob}. 
Define \smash{$z=\alpha \hbeta_1+ (1-\alpha)\hbeta_2$} for any
$0<\alpha<1$. 
Then by strict convexity of $f$ and convexity of $\pen$, 
\begin{align*}
f(Xz)+\lambda \pen(z) &< \alpha f(X\hbeta_1) + (1-\alpha)
f(X\hbeta_2) + \lambda \Big(\alpha \pen(\hbeta_1) + (1-\alpha) 
\pen(\hbeta_2) \Big) \\ 
&= (1-\alpha) c^* + \alpha c^* \\
&= c^*,
\end{align*}
contradicting the fact that $c^*$ is the optimal value of the
criterion.  Therefore we conclude that \smash{$X\hbeta_1=X\hbeta_2$}.

Furthermore, the uniqueness of the fitted value $X\hbeta$ in
\eqref{eq:genprob} implies uniqueness of the penalty term
\smash{$\pen(\hbeta)$}, for any $\lambda > 0$. 

\subsection{Calculation of $\lambda_1$}
\label{app:genlam1}

By the subgradient optimality conditions, $\hbeta$ is a solution in
\eqref{eq:genprob} if and only if
\begin{equation*}
\frac{X^T (y-X\hbeta)}{\lambda} \in \partial \pen(\hbeta),
\end{equation*}
where \smash{$\partial \pen(\hbeta)$} denotes the subdifferential (set of
subgradients) of $\pen$ evaluated at \smash{$\hbeta$}. Recalling that 
$\pen(\beta) = \max_{u\in\dualball} u^T \beta$, we can rewrite
this as
\begin{equation}
\label{eq:sg}
\frac{X^T (y-X\hbeta)}{\lambda} \in \dualball \qquad \text{and} \qquad
\frac{(X\hbeta)^T(y-X\hbeta)}{\lambda} = \pen(\hbeta).
\end{equation}
The first statement above is equivalent to
\smash{$\dualpen(X^T(y-X\hbeta)) \leq \lambda$}, where
$\dualpen(\beta) = \max_{v\in\dualball^\circ} v^T \beta$ and  
$\dualball^\circ$ is the polar body of $\dualball$.  Now define  
\begin{equation*}
\lambda_1 = \dualpen\big(X^T (I-P_{X\dualball^\perp}) y\big), 
\end{equation*}
where $P_{X\dualball^\perp}$ denotes projection onto the linear
subspace 
\begin{equation*}
X\{v : v \perp \dualball\} = X \dualball^\perp.
\end{equation*}
Then for any $\lambda \geq \lambda_1$, the conditions in
\eqref{eq:sg} are satisfied by taking
\smash{$X\hbeta=P_{X\dualball^\perp}y$}, and 
accordingly, \smash{$\pen(\hbeta)=0$}.  

On the other hand, if $\lambda < \lambda_1$, then we must have
\smash{$\pen(\hbeta) \not= 0$}, because \smash{$\pen(\hbeta)=0$}
implies that \smash{$X\hbeta \in X \dualball^\perp$}, in fact
\smash{$X\hbeta=P_{X\dualball^\perp} y$}, so  $\lambda < \dualpen(X^T
(I-P_{X\dualball^\perp})y)$ and the first condition in \eqref{eq:sg}
is not satisfied.  

\subsection{Bounded support function}
\label{app:welldefined}

We first establish a technical lemma.

\begin{lemma}
\label{lem:boundedpolar}
Suppose that $D \subseteq \real^n$ has $0 \in \relint(D)$.  
Then for all $u \in \spa(D)$, 
\begin{equation*}
\max_{v \in D^\circ}\, v^Tu \leq M,
\end{equation*}
for some constant $M < \infty$, with
$D^{\circ} = \{v: u^Tv \leq 1 \;\, \text{for all}\;\, u \in D\}$
the polar body of $D$. In other words, $D^\circ \cap \spa(D)$ is a
bounded set.
\end{lemma}
\begin{proof}
By assumption, we know that $P_D B(r) \subseteq D$
for some $r > 0$, where $P_D$ is the projection matrix onto
$\spa(D)$, and $B(r)$ is the $\ell_2$ ball of radius $r$
centered at 0. Fix $\delta>0$, and define $S_{D,r^{-1}+\delta}
= \partial P_D B(r^{-1} + \delta)$, the boundary of the projected 
ball of radius $r^{-1}+\delta$. Then for any
$v\in S_{D,r^{-1}+\delta}+D^\perp$, 
we can define $u=r P_D v / \|P_D v\|_2$, and we have $u \in P_D
B(r) \subseteq D$ with $u^T v = 1+\delta r > 1$, hence $v
\notin D^\circ$.  As this holds for all $\delta > 0$, we must have 
\begin{equation*}
D^\circ \subseteq \overline{\real^p \setminus 
\Big(\bigcup_{\delta > 0} S_{D,r^{-1}+\delta} +
D^\perp \Big)}= P_D B(r^{-1}) +
D^{\perp}.
\end{equation*}
This implies that $P_D D^{\circ} \subseteq P_D B(r^{-1})$,
which gives the result.
\end{proof}

We now use Lemma \ref{lem:boundedpolar} to show that the process
$f_\eta$ in \eqref{eq:f} is bounded over $\eta\in\K=\dualball^\circ$,
assuming that $\dualball$ is closed and contains 0 in its relative
interior. 

\begin{lemma}
\label{lem:welldefined}
If $\dualball \in \real^p$ is a closed, convex set with $0 \in
\relint(\dualball)$, then 
\begin{equation*}
\max_{\eta\in\K}\, f_\eta \leq M,
\end{equation*}
where $f_\eta$ is the process in \eqref{eq:f}, $\K=\dualball^\circ$,
and $M<\infty$ is a constant.  
\end{lemma}
\begin{proof}
We reparametrize $f_\eta = (X\eta)^T(I-P_{X\dualball^\perp})y$ as
\begin{equation*}
g_v = v^T (I-P_{X\dualball^\perp})y, \qquad v \in X\K,
\end{equation*} 
and apply Lemma \ref{lem:boundedpolar} to the process $g_v$, with 
$D = (X^T)^{-1} \dualball$, the inverse image of $\dualball$ under
the linear map $X^T$.  It is straightforward to check, using the fact
that $\dualball$ is closed, convex, and contains 0 (i.e., using
$\dualball^{\circ\circ}=\dualball$), that $D^\circ = X\dualball^\circ
= X\K$.  By Lemma \ref{lem:boundedpolar}, therefore, we know that 
$g_v \leq M$ for all $v \in \spa(D)$.  The proof is completed by
noting that $\spa(D) = \spa((X^T)^{-1} \dualball) =
(X\dualball^\perp)^\perp$, which means
$(I-P_{X\dualball^\perp})y \in \spa(D)$ for any $y \in \real^n$.  
\end{proof}

\bibliographystyle{agsm}
\bibliography{kacricetest}

\end{document}